\documentclass[letterpaper,11pt]{article}
\usepackage[utf8]{inputenc}

\pdfoutput=1

\usepackage{amsmath, amsthm, amssymb, thm-restate}
\usepackage{algorithmicx}
\usepackage[table,xcdraw]{xcolor}
\usepackage[ruled,vlined,linesnumbered]{algorithm2e}
\usepackage{setspace}
\usepackage{mathtools}
\usepackage[style=alphabetic,natbib=true,maxbibnames=9, maxcitenames=9]{biblatex}

\usepackage{comment}
\usepackage[most]{tcolorbox} 
\usepackage{xfrac}
\usepackage{hyperref}
\usepackage{multirow}
\usepackage{caption}
\usepackage{bm}
\usepackage{newfloat}
\usepackage{enumitem}
\usepackage{dblfloatfix} 
\usepackage{wrapfig}

\usepackage{fancyhdr}

\usepackage[margin=1in]{geometry}

\allowdisplaybreaks

\definecolor{mygreen}{RGB}{10,150,110}
\definecolor{myred}{RGB}{150,10,20}

\hypersetup{
     colorlinks=true,
     citecolor= mygreen,
     linkcolor= myred
}

\renewcommand{\epsilon}{\varepsilon}

\DeclareMathOperator{\E}{\ensuremath{\normalfont \textbf{E}}}

\newcommand{\hiddencomment}[1]{}

\newcommand{\mc}[1]{\ensuremath{\mathcal{#1}}}

\newcommand{\yes}[0]{\ensuremath{\mathsf{YES}}}
\newcommand{\no}[0]{\ensuremath{\mathsf{NO}}}
\newcommand{\yesdist}[0]{\ensuremath{\mathcal{D}_{\yes}}}
\newcommand{\nodist}[0]{\ensuremath{\mathcal{D}_{\no}}}
\newcommand{\dist}[0]{\ensuremath{\mathcal{D}}}

\newcommand{\smparagraph}[1]{\vspace{-0.2cm}\paragraph{#1}}

\DeclareMathOperator{\poly}{poly}

\usepackage[noabbrev,nameinlink]{cleveref}
\crefname{lemma}{Lemma}{Lemmas}
\crefname{theorem}{Theorem}{Theorems}
\crefname{property}{Property}{Properties}
\crefname{claim}{Claim}{Claims}
\crefname{result}{Result}{Results}
\crefname{definition}{Definition}{Definitions}
\crefname{observation}{Observation}{Observations}
\crefname{proposition}{Proposition}{Propositions}
\crefname{assumption}{Assumption}{Assumptions}
\crefname{line}{Line}{Lines}
\crefname{figure}{Figure}{Figures}
\creflabelformat{property}{(#1)#2#3}
\crefname{equation}{}{}
\crefname{section}{Section}{Sections}
\crefname{appendix}{Appendix}{Appendices}
\crefname{algCounter}{Algorithm}{Algorithms}
\Crefname{algCounter}{Algorithm}{Algorithms}

\newtheorem{theorem}{Theorem}

\newtheorem{lemma}{Lemma}[section]
\newtheorem{proposition}[lemma]{Proposition}
\newtheorem{corollary}[lemma]{Corollary}

\newtheorem{definition}[lemma]{Definition}
\newtheorem{claim}[lemma]{Claim}

\newtheorem{observation}[lemma]{Observation}
\newtheorem{remark}{Remark}
\newtheorem*{remark*}{Remark}

\definecolor{mylightgray}{RGB}{230,230,230}

\algnewcommand{\IIf}[2]{\textbf{if} #1 \textbf{then} #2}
\algnewcommand{\EndIIf}{\unskip\ \algorithmicend\ \algorithmicif}

\newenvironment{whitetbox}{
\par\addvspace{0.1cm}
\begin{tcolorbox}[width=\textwidth,
                  boxsep=5pt,
                  left=1pt,
                  right=1pt,
                  top=2pt,
                  bottom=2pt,
                  boxrule=1pt,
                  arc=0pt,
                  colframe=black,
                  colback=white
                  ]
}{
\end{tcolorbox}
}

\newcounter{algCounter}

\makeatletter
\renewcommand{\paragraph}{%
  \@startsection{paragraph}{4}%
  {\z@}{10pt}{-1em}%
  {\normalfont\normalsize\bfseries}%
}
\makeatother

\makeatletter
\patchcmd{\@algocf@start}
  {-1.5em}
  {0pt}
  {}{}
\makeatother

\title{Approximating Maximum Matching\\ Requires Almost Quadratic Time}

 \author{
  Soheil Behnezhad\\{\em Northeastern University} \and
  Mohammad Roghani\\{\em Stanford University} \and
  Aviad Rubinstein\\{\em Stanford University}}

\date{}

\begin{document}

\maketitle

\thispagestyle{empty}

\begin{abstract}
    We study algorithms for estimating the size of maximum matching. This problem has been subject to extensive research. For $n$-vertex graphs, Bhattacharya, Kiss, and Saranurak [FOCS'23] (BKS) showed that an estimate that is within $\epsilon n$ of the optimal solution can be achieved in $n^{2-\Omega_\epsilon(1)}$ time, where $n$ is the number of vertices. While this is subquadratic in $n$ for any fixed $\epsilon > 0$, it gets closer and closer to the trivial $\Theta(n^2)$ time algorithm that reads the entire input as $\epsilon$ is made smaller and smaller.

    \smallskip\smallskip

    In this work, we close this gap and show that the algorithm of BKS is close to optimal. In particular, we prove that for any fixed $\delta > 0$, there is another fixed $\epsilon = \epsilon(\delta) > 0$ such that estimating the size of maximum matching within an additive error of $\epsilon n$ requires $\Omega(n^{2-\delta})$ time in the adjacency list model.
\end{abstract}

{
\clearpage
\hypersetup{hidelinks}
\vspace{1cm}
\renewcommand{\baselinestretch}{0.1}
\setcounter{tocdepth}{2}
\clearpage
}

\setcounter{page}{1}

\section{Introduction}

The maximum matching problem in graphs has been a cornerstone in theoretical computer science, with a rich history spanning several decades. A {\em matching} is a set of vertex-disjoint edges. A {\em maximum matching} is a matching that is largest in size. In graphs with $n$ vertices and $m$ edges, a maximum matching can be found in $O(m\sqrt{n})$ time \cite{MicaliV80}. If we only desire $(1-\epsilon)$-approximations of maximum matching\footnote{See \cref{sec:preliminaries} for the formal definition of approximate maximum matchings.\label{fn:approx}}, then the running time improves to $O(m/\epsilon)$ which is linear in the input size \cite{MicaliV80,DuanP14}. However, the modern landscape of graph analysis often involves dealing with graphs of monumental scale, rendering even linear-time algorithms impractically slow for many applications. This motivates the study of {\em sublinear time} algorithms whose goal is to derive approximations of the maximum matching without a full traversal of the entire graph.

Indeed, the complexity of approximating the maximum matching size in sublinear time has been an area of intense study that has led to numerous breakthroughs over the years.  We first overview these existing bounds and then discuss our contribution in this paper.

When considering sublinear time algorithms, it is important to first specify how the input can be accessed. For graphs, two models are most common: the adjacency list model and the adjacency matrix model. Our focus in this work is on the former. In this model, each query of the algorithm specifies a vertex $v$ and an integer $i$. The response is the $i$-th neighbor of $v$ in an arbitrarily ordered list or $\perp$ be $v$ has fewer than $i$ neighbors.

\vspace{-0.2cm}
\paragraph{Related work:} Earlier works on sublinear-time algorithms for maximum matching focused on graphs of bounded degree $\Delta$. This was pioneered by \citet*{ParnasRon07} who gave an algorithm with a quasi-polynomial in $\Delta$ running time of $\Delta^{O(\log \Delta)}$ that estimates the size of maximum matching up to a multiplicative-additive\textsuperscript{\ref{fn:approx}} factor of $(1/2, \epsilon n)$. The dependency on $\Delta$ was later improved to polynomial. Building on the randomized greedy approach of \citet*{NguyenOnakFOCS08}, it was shown by \citet*{YoshidaYISTOC09} that a $(1, \epsilon n)$-approximation can be achieved in $\Delta^{O(1/\epsilon^2)}$ time. Whether this can be further improved to $\poly(\Delta/\epsilon)$ remained open until a recent work of  \citet*{BehnezhadRR-FOCS23} ruled it out. In particular, they showed that $\Delta^{\Omega(1/\epsilon)}$ time is needed for obtaining a $(1, \epsilon n)$-approximation \cite{BehnezhadRR-FOCS23}.

The above-mentioned algorithms do not run in sublinear time in general graphs where $\Delta$ can be as large as $\Omega(n)$. There has also been a long line of work on achieving algorithms with subquadratic-in-$n$ (and thus sublinear in the input size which can be $\Omega(n^2)$) running times \cite{KapralovSODA20,ChenICALP20,behnezhad2021,BehnezhadRRS-SODA23,Behnezhad-RRS-STOC23,BhattacharyaKS-STOC23,BhattacharyaKS-FOCS23} in general graphs. For instance, \citet*{behnezhad2021} showed a $(1/2 - \epsilon)$-approximation can be obtained in $\widetilde{O}(n)$ time. After a series of improvements over the approximation ratio \cite{BehnezhadRRS-SODA23,Behnezhad-RRS-STOC23,BhattacharyaKS-STOC23,BhattacharyaKS-FOCS23}, \citet*{BhattacharyaKS-FOCS23} showed that a $(1, \epsilon n)$-approximation can be obtained in $n^{2-\Omega_\epsilon(1)}$ time.\footnote{We note that the result of \cite{BhattacharyaKS-FOCS23} is stated in the adjacency matrix model. However, their algorithm is also believed to extend to the adjacency list model achieving a $(1, \epsilon n)$ approximation in $n^{2-\Omega_\epsilon(1)}$ time \cite{BKSpersonalcom}.} While this is subquadratic in $n$ for any fixed $\epsilon > 0$,  as $\epsilon$ diminishes, its runtime gets close to $\Omega(n^2)$.

On the lower bound side, the situation is very different.
Sixteen years ago, \citet*{ParnasRon07} proved that achieving a constant approximation of the maximum matching size requires at least $\Omega(n)$ time.
The authors \cite{Behnezhad-RRS-STOC23} showed that any algorithm achieving a $(2/3+\Omega(1), \epsilon n)$-approximation requires at least $n^{6/5-o(1)}$ time. For sparse graphs, \cite{BehnezhadRR-FOCS23} prove a lower bound of $\Delta^{\Omega(1/\epsilon)}$ for $(1, \epsilon n)$-approximation; their construction can be carefully adapted to show a lower bound of $n^{3/2-\delta(\epsilon)}$ time for dense graphs, but as we discuss in \cref{sec:techniques} there is a major barrier for extending it beyond $n^{1.5}$, regardless of $\Delta$. This has left out the possibility of an algorithm running in time as small as $n^{1.5}\poly(1/\epsilon)$ and achieving a $(1, \epsilon n)$-approximation. Whether such extremely fast algorithms exist has remained open.

\vspace{-0.2cm}
\paragraph{Our contribution:} In this paper, we close this huge gap by showing that the algorithm of \citet{BhattacharyaKS-FOCS23} is close to optimal. That is, we present a new lower bound that shows near-quadratic in $n$ time is necessary in order to achieve a $(1, \epsilon n)$-approximation of the maximum matching size. \cref{thm:main} below is the formal statement of our lower bound.

\begin{tcolorbox}[parbox=false]
\vspace{-0.5cm}
    \begin{theorem}[\textbf{Main Result}]\label{thm:main}
        For any $\delta > 0$ there is an $\epsilon = \epsilon(\delta) > 0$  such that any (randomized) algorithm that (with probability at least $2/3$) estimates the size of maximum matching of an $n$-vertex graph up to an additive error of $\epsilon n$ has to make $\Omega(n^{2-\delta})$ adjacency list queries to the graph.

        Furthermore, this holds even if the graph is bipartite and is promised to either have a perfect matching or a matching that leaves $\Theta(\epsilon n)$ vertices unmatched.
    \end{theorem}
\end{tcolorbox}

The prior lower bound analysis of \cite{Behnezhad-RRS-STOC23,BehnezhadRR-FOCS23} work in a certain {\em tree model} and rely crucially on the fact that the algorithm cannot discover any cycles.\footnote{More precisely, the construction of \cite{Behnezhad-RRS-STOC23} has $\epsilon n$ {\em dummy} vertices that are adjacent to all the rest of vertices which are called the {\em core} vertices. The assumption in \cite{Behnezhad-RRS-STOC23} is that the algorithm cannot find any cycles in the graph induced by the core vertices.} It turns out that this assumption completely breaks when the algorithm is allowed to make $\omega(n\sqrt{n})$ queries. This is the main conceptual and technical obstacle that our lower bound of \cref{thm:main} overcomes. In \cref{sec:techniques}, we elaborate more on the cycle discovery barrier, its importance in the literature of sublinear time algorithms and lower bounds, and our techniques to bypass it for approximating maximum matchings.

\vspace{-0.2cm}
\paragraph{Paper organization:} We present an overview of our techniques in \cref{sec:techniques}. In \cref{sec:preliminaries}, we formalize the notation and definitions we use and provide the needed background. After presenting a table of the parameters we use in \cref{sec:tableofparameters}, we formalize our input construction in \cref{sec:input-distribution}. Finally, we prove our lower bound of \cref{thm:main} in \cref{sec:indistinguishability,sec:finalizing-the-main-theorem}. That is, we show that no algorithm that makes $O(n^{2-\delta})$ queries can distinguish whether our input construction contains a perfect matching or its maximum matching leaves $\Omega(\epsilon n)$ vertices unmatched.

\subsection{Further Related Work: Dynamic Algorithms}

Besides being an important problem on its own, the study of sublinear time algorithms for maximum matching has also recently found applications in the dynamic setting \cite{Behnezhad-SODA23,bhattacharyaKSW-SODA23,BhattacharyaKS-FOCS23,AzarmehrBR-SODA24}. In this setting, the graph undergoes a sequence of edge insertions and deletions, and the goal is to maintain (the size of) a maximum matching efficiently after each update.

The connection is as follows. Suppose we have a sublinear time algorithm that estimates the maximum matching size within a factor of $(\alpha, \epsilon n)$ in $T$ time. Then in the dynamic setting, we can only call this sublinear time algorithm after every $\epsilon n$ updates. Since the maximum matching size changes by at most 1 after each update, this remains a $(\alpha, 2\epsilon n)$ estimation throughout all updates. Additionally, the amortized update-time is now $O(T/\epsilon n)$.

Based on this connection, the $n^{2-\Omega_\epsilon(1)}$ time sublinear-time algorithm of \citet*{BhattacharyaKS-FOCS23} leads to a $(1, \epsilon n)$-approximation of maximum matching size in $n^{1-\Omega_\epsilon(1)}$ amortized time per update, polynomially breaking the linear-in-$n$ barrier for the first time. A natural next question is can we obtain a $(1, \epsilon n)$-approximation much faster, in say, $n^{0.9} \poly(1/\epsilon)$ amortized update time? Our lower bound of \cref{thm:main} shows this framework of using sublinear time algorithms as black-box cannot lead to such running times.

It is worth noting that the complexity of the sublinear matching problem presents a barrier for the sublinear time algorithm for the path cover problem \cite{TSP-icalp24}, which has applications in the sublinear time algorithm for estimating the traveling salesman problem (TSP), which has been studied in the literature of sublinear time algorithms \cite{TSP-icalp24, ChenMetric-Arxiv22, ChenICALP20}. For further details on the connection between these two problems, we encourage readers to refer to Section 11 of \cite{TSP-icalp24}.

\section{Our Techniques: Bypassing The Cycle Discovery Threshold}\label{sec:techniques}

In this section, we provide a high level overview of our lower bound of \cref{thm:main}. 

As already discussed, existing lower bounds \cite{Behnezhad-RRS-STOC23,BehnezhadRR-FOCS23} rely heavily on inability of efficient algorithms to discover certain cycles. This assumption completely breaks when the algorithm is allowed to make $\omega(n\sqrt{n})$ queries. Our contribution in this work is to break this cycle discovery threshold, showing that even though an algorithm with near-quadratic queries can discover cycles, it cannot estimate the size of maximum matching.  

We start by providing some background on the existing lower bounds, then discuss the cycle discovery barrier in more detail, and finally overview our new ideas to bypass it.

\subsection{Background on Existing Lower Bounds}

\paragraph{High degree dummy vertices.}
The first basic idea for proving query complexity lower bounds in the adjacency list model, also common in earlier lower bounds \cite{ParnasRon07,BehnezhadRRS-SODA23}, is to add $\epsilon n$ {\em dummy} vertices and make them adjacent to the rest of vertices. The dummy vertices do not contribute significantly to the maximum matching as there are few of them, but increase the number of edges to $\Omega(\epsilon n^2)$, effectively congesting the adjacency lists with redundant calls. We henceforth refer to the non-dummy part of the graph as the {\em core}. That is, non-dummy vertices are {\em core vertices} and edges between core vertices are {\em core edges}. 

\citet*{ParnasRon07} gave a linear lower bound of $\Omega(n)$ for any constant approximate algorithm by taking the core to be (essentially) either a random perfect matching or the empty graph. Intuitively, because of the dummy vertices, it takes the algorithm $\Omega(n)$ adjacency list queries to even hit one edge of the matching edges in the core. This argument breaks when the goal is to prove super-linear lower bounds. Note that if the algorithm is able to make $nk$ queries, then it can random sample $k$ vertices and query all of their adjacency lists, therefore at least $\Omega(k)$ edges of the maximum matching of size $\Theta(n)$ will be revealed to the algorithm.

\smparagraph{Camouflage the good matching.} As discussed above it is impossible to hide the maximum matching edges in the sense that some of them will be revealed to the algorithm. The approach pioneered by the work of \citet*{Behnezhad-RRS-STOC23} to overcome this challenge is to introduce a special construction which {\em camouflages} the edges of the maximum matching, in the sense that they are statistically indistinguishable to the algorithm from the rest of the edges in the core that do not participate in a maximum matching. This is the key feature of the new construction in~\cite{Behnezhad-RRS-STOC23}  that obtains the first super-linear query lower bound of $\Omega(n^{1.2})$ for approximating maximum matching.

In a little more detail, it was shown in \cite{Behnezhad-RRS-STOC23} that so long as the average degree in the core is not too large (say smaller than $n^{0.2}$) and the algorithm does not conduct too many queries (say smaller than $n^{1.2}$), then the discovered edges of the core will form a forest. This enables \cite{Behnezhad-RRS-STOC23} to argue that the algorithm cannot distinguish the edges of the maximum matching from the rest of core edges by reducing the problem to a label guessing game on trees. 

\subsection{The Cycle Discovery Barrier}
 The assumption that the algorithm cannot discover any cycles in the core completely breaks when the algorithm is allowed to make $\omega(n^{1.5})$ queries, making it particularly challenging to prove such lower bounds. To provide some intuition about this, suppose that we take a vertex $v$ and run a BFS from it (discarding dummy vertices) until discovering $\Theta(\sqrt{n})$ vertices of the core. Note that this takes only $O(n\sqrt{n})$ queries even if we query the whole adjacency list of each encountered core vertex. Informally speaking, if we run this BFS from two random starting vertices, then by the birthday paradox, we expect their discovered descendants to collide, therefore forming cycles. 

At first glance, this may seem like a limitation of existing lower bounds proofs rather than a strength of these algorithms. However, we remark that there is indeed an algorithm running in $\widetilde{O}(n\sqrt{n})$ time that solves the construction of \cite{Behnezhad-RRS-STOC23}. It is also worth noting that the cycle discovery threshold does indeed represent the correct bound for other problems in the sublinear time model. For instance, \citet*{GoldreichR-STOC97} first gave a lower bound of $\Omega(\sqrt{n})$ for bipartite testing, using also the assumption that faster algorithms cannot discover cycles. Later, in a follow up work, they showed that there is indeed an algorithm running in $\widetilde{O}(\sqrt{n})$ time for this problem \cite{GoldreichR-STOC98}.

To recap, the approach in previous work \cite{Behnezhad-RRS-STOC23} was to camouflage the edges of the good matching. The limitation of  the previous approach is that once the algorithm gets $n^{1.5}$ queries, it can discover at least $\sqrt{n}$ core edges, at which point the algorithm discovers cycles. And cycles break the camouflage of the edges on the cycle.

\subsection{Our Key Contribution: Bypassing the Cycle Discovery Barrier}\label{sec:cycle-barrier}

Our key novel idea in this work to bypass the cycle discovery barrier is to camouflage the entire core instead of just the maximum matching. How do we camouflage the entire core? Roughly the same way that previous work camouflaged the good matching! This (in hindsight) inspires our construction: we have a recursive construction of $L$ {\em levels}; the $i$-th level is similar to the entire construction of~\cite{BehnezhadRR-FOCS23}, with the main difference being that we replace the hidden good matching with the $(i-1)$-level construction. To provide more details, let us first overview the base of the construction due to \cite{BehnezhadRR-FOCS23}. 

\paragraph{The base (due to \cite{BehnezhadRR-FOCS23}):}  Consider the graph illustrated below with $2r+1$ subsets of vertices $A_1, \ldots, A_r$, $B_1, \ldots, B_r$, and $S$, where $r$ is a parameter of the construction (it is instructive to take $r=1/\epsilon$). For any $i \in [r]$, there is an $n^{2\epsilon}$-regular bipartite graph between $A_i$ and $B_i$ that we call a block. There is a perfect matching from $A_i$ to $B_{i+1}$, a perfect matching from $S$ to $B_1$, and a perfect matching from $A_r$ to $A_r$ (which may or may not exist). We call the edges of these perfect matchings {\em special} edges.

\begin{figure}[h]
    \centering
    \includegraphics[scale=0.7]{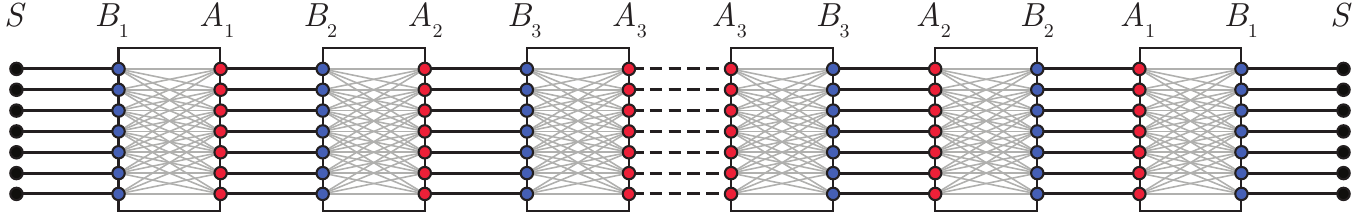}
\end{figure}

It is not hard to see that any algorithm achieving better than a $(1-\frac{1}{2r+1})$ approximation must verify whether the $A_r$-$A_r$ matching exists. Therefore, it suffices to show that near quadratic queries are needed to determine this.\footnote{We remark that in our final construction, we ensure that the $A_r$ vertices have the same degrees as vertices in other layers no matter whether the $A_r$-$A_r$ matching exists. We hide these details here for our informal overview of our lower bound.} To do so, we would like to argue a vertex $v$ does not know which of its edges are special, thus it has to do a BFS of depth $r=1/\epsilon$ to reach the $S$ vertices, exploring $\Omega((n^{2\epsilon})^{r}) = \Omega(n^2)$ edges. The problem, however, is exactly the cycle discovery problem. The birthday-paradox argument discussed earlier can be used to test in $O(n^{1.5-2\epsilon})$ time whether two vertices $u$ and $v$ are in the same block. Therefore, a vertex can find its special edge in just $O(n^{1.5-2\epsilon}) \cdot O(n^{2\epsilon}) = O(n^{1.5})$ time by running this test on all of its $n^{2\epsilon}$ neighbors. Continuing along these special edges, we reach $S$ in just $O(r n^{1.5}) = O_\epsilon(n^{1.5})$ time overall.

\paragraph{The recursion (new to this work):} To resolve this problem, we give a recursive construction. In particular, we will define a sequence of input constructions denoted as $G^1, \ldots, G^L$, where $G^1$ is (essentially) the construction discussed above. The construction of $G^\ell$ is the same as our construction for $G^1$, except that we replace the special edges (i.e., the perfect matchings) with the graph of the previous level $G^{\ell-1}$. The figure below illustrates this.

\begin{figure}[h]
    \centering
    \includegraphics[scale=0.7]{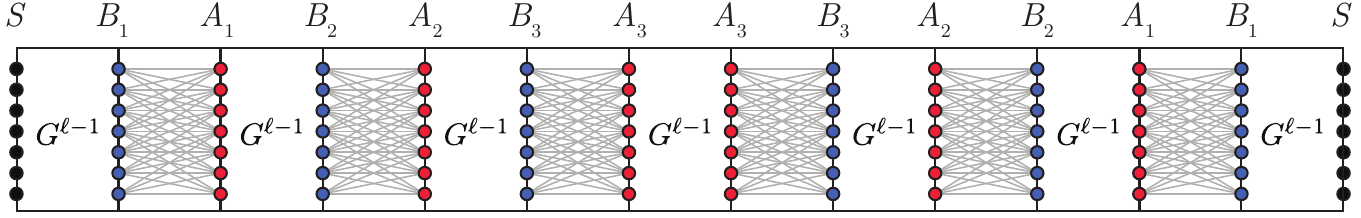}
    \caption{Construction of $G^\ell$ based on $G^{\ell-1}$.}
\end{figure}

Let us use $n^{2-\delta}$ to denote the number of queries that the algorithm makes and use $n^{\sigma}$ to denote the degrees in the regular blocks of graph $G^\ell$. The key to our analysis is to show that while we discover {\em some} cycles which ``spoil'' the camouflage of some edges, the number of cycles, and hence also the number of spoiled edges, decreases by an $n^{\delta - \sigma}$ factor at each level. With a sufficiently large number of levels, we can ensure that the total decrease in the number of spoiled edges is significant. This ensures that at the bottom level $G^1$, we cannot discover any cycles at all. Consequently, we can safely camouflage the edges of the good matching and prove our lower bound using the previously known techniques when there exists no cycle. Throughout the remainder of this technical overview, our main focus is to provide a high-level intuition for why this decrease in advantage occurs when we move one level down in the construction.

\subsection{Decrease in the Advantage of the Algorithm in Discovering Camouflaged Edges}

To understand this decrease in the algorithm's advantage in detecting camouflaged edges, examine the highest level in the recursive construction, denoted as $G_L$. For the remainder of this section, we say that an edge discovered by the algorithm is directed from $u$ to $v$ if it is obtained by querying the adjacency list of $u$.

We claim that each vertex has a probability of $O(1/n)$ to be the answer to each adjacency list query. First, note that the gadgets that we are using in our construction are random regular graphs. Consider a pair of vertices $(u, v)$ in the core construction that is not among the edges discovered by the algorithm. Let $x$ and $y$ be the number of undiscovered core edges of $u$ and $v$, respectively. Using a coupling argument, we can show that there exists an edge between $u$ and $v$ with a probability $O(\min(x, y)/n)$ (see \Cref{lem:edge-prob-bound-remaining}). Combining the above argument and the fact that the adjacency list of vertices is ordered uniformly at random implies that when the algorithm queries the adjacency list of vertex $u$, given that this vertex has $x$ undiscovered edges, the probability of the answer to this query being a specific vertex $v$ is bounded by $O(1/n)$ (see \Cref{clm:incoming-prob}). Applying this observation, we obtain several properties of the subgraph queried by the algorithm. We mention a few of these properties that are useful in our proof:
\begin{itemize}
    \item[(P1)] \textbf{Each vertex in the core has $O(\log n)$ incoming edges:} Consider a vertex $v$ in the core. If we query the adjacency list of vertex $u$ in the core, the probability of obtaining a directed edge $(u,v)$ is bounded by $O(1/n)$. Given that a fraction of $O(n^\sigma/n)$ edges of the whole input graph are in the core, the algorithm is going to discover at most $O(n^{1-\delta+\sigma}) \ll O(n)$ edges of the core (see \Cref{clm:discovered-edges-bound}). Hence, the expected number of incoming edges for each vertex is less than one. Using a concentration inequality, we can show that with high probability, $v$ has at most $O(\log n)$ incoming edges (see \Cref{clm:max-in-degree}).
    
    \item[(P2)] \textbf{Most edges in the core do not close a cycle:} As discussed in (P1), the algorithm can discover at most $O(n^{1-\delta+\sigma})$ edges of the core. Therefore, at any time during the execution of the algorithm, there are at most $O(n^{1-\delta+\sigma})$ vertices with at least one edge in the core. This implies that the probability that the answer to each new query made by the algorithm is a vertex for which the algorithm has previously found an incident edge in the core is $O(n^{\sigma - \delta})$. This suggests that the majority of edges in the core do not close a cycle, with only a fraction of $O(n^{\sigma - \delta})$ closing cycles.
    
    \item[(P3)] \textbf{Local directed neighborhood of most vertices in core is a small tree:} For a vertex $v$ in core, let the {\em shallow subgraph} of $v$, denoted $T(v)$, be the set of vertices that are reachable from $v$ using queried directed paths of length at most $O(\log n)$ with edges in the core. Now consider all the edges in the core. Using a stronger argument similar to (P1), we can show that each queried edge by the algorithm belongs to at most $\poly\log(n)$ different shallow subgraphs. Therefore, we have $\sum_v |T(v)| \leq \widetilde{O}(n^{1-\delta + \sigma})$. We say a shallow subgraph is {\em small} if it has less than $n^{\delta - 2\sigma}$ vertices, and {\em large} otherwise. By our bound on $\sum_v |T(v)|$, we can have at most $\widetilde{O}(n^{1 - 2\delta + 3\sigma})$ large shallow subgraphs. On the other hand, only $O(n^{\sigma - \delta})$-fraction of the core edges close a cycle by (P2), aka  there are at most $O(n^{1-2\delta + 2\sigma})$ such edges. Consequently, there are at most $\widetilde{O}(n^{1-2\delta + 2\sigma})$ shallow subgraphs that have a cycle. As a result, the local directed neighborhood of most of the vertices is a tree of size $O(n^{\delta - 2\sigma})$.
\end{itemize}

For formal proof of (P2) and (P3), we encourage readers to see \Cref{lem:spoiled-vertices-bound}. Suppose that we define vertex labels similarly as discussed in \Cref{sec:cycle-barrier}, i.e. vertices can have labels $A_1,\ldots, A_r,$ 
$ B_1, \ldots, B_r$. For a vertex with property (P3), referred to as an {\em unspoiled} vertex, we can demonstrate that the algorithm is incapable of distinguishing the vertex's label. The technical proof for this part is mostly borrowed from \cite{BehnezhadRR-FOCS23} and uses the fact that in each level of our construction, the graph is very similar to the construction of \cite{BehnezhadRR-FOCS23}. Finally, we can argue that for all edges with unspoiled endpoints, the algorithm has a negligible bias in the probability of the edge belonging to a lower level (gadget between $A_r$ and $A_r$). More formally, considering the bound obtained in (P3), there are at most $n^{1-2\delta + 4\sigma}$ spoiled vertices. Consequently, there are at most $n^{1-2\delta + 5\sigma}$ edges for which the algorithm has a significant bias in the probability that they belong to a lower level (see \Cref{lem:prob-of-edge-being-black}). We encourage readers to refer to the warm-up presented in \Cref{sec:warm-up}, as many of the ideas mentioned here are discussed in more detail there, and it contains many key ideas essential to our proof.

Assume that the degree of vertices for inner level $G^{L-1}$ are asymptotically smaller, i.e. $n^{\sigma'}$ where $\sigma'< \sigma$. Exclude the edges—amounting to $n^{1-2\delta+5\sigma}$—that the algorithm distinctly identifies as belonging to a lower level due to a significant bias in probability. For all remaining edges, due to the minor bias, the probability of the edge belonging to level $L-1$ is at most $O(n^{\sigma' - \sigma})$. Intuitively, this implies that a majority of the edges belong to a higher level, and the algorithm is unable to form large connected components of the inner level using unbiased edges. To observe this contrast in the size of connected components, consider the following simple and intuitive example. At the highest level, the algorithm can concentrate all its queries to create a single large component of size $n^{1-\delta + \sigma}$. Now, let us suppose the algorithm is executing a BFS from an arbitrary vertex in the graph to create large components of inner edges using unbiased edges. In each step, the algorithm queries all neighbors during BFS. It is noteworthy that each edge belongs to the inner level with a probability of $O(n^{\sigma' - \sigma})$. Consequently, the size of the largest component with inner edges is $O(n^{(\sigma'/\sigma)(1 - \delta + \sigma)}) \ll O(n^{1 - \delta + \sigma})$. The decrease in the size of the connected components aids in demonstrating that, in the lower level, the count of vertices proximate to cycles is considerably smaller. By recursively applying this step, ultimately, we can reach the base level where we can prove, with high probability, the absence of cycles.

We point out that the informal outline above oversimplifies several important parts of our proof. Firstly, the construction discussed above as stated can be solved efficiently with a random-walk based argument. To resolve this, we add a number of {\em delusive} vertices (introduced before by \cite{BehnezhadRR-FOCS23}) to each level of the recursion where roughly speaking $\epsilon$ fraction of edges of each vertex go to these delusive vertices. Secondly, the degrees of the regular blocks and the number of blocks in each level of the recursion have to be balanced carefully. In particular, we need to ensure that the blocks in $G^{\ell}$ are sufficiently denser than those in $G^{\ell-1}$ to be able to argue that we see fewer cycles in $G^{\ell-1}$ than in $G^\ell$. But having smaller degrees in $G^{\ell-1}$ requires increasing the number of blocks in $G^{\ell-1}$ to keep it essentially as ``difficult'' to solve as $G^\ell$. Finally, the queries conducted at level $G^\ell$ reveal some information about the labels in the previous level $G^{\ell-1}$. This has to be quantified carefully in order to formalize the intuitive argument that the algorithm sees fewer cycles in $G^{\ell-1}$.

\newpage
\section{Preliminaries}\label{sec:preliminaries}

\paragraph{Notation:} Throughout this paper, we use $G = (V, E)$ to denote the input graph. Moreover, we use $n$ to denote the number of vertices in $G$, $\mu(G)$ to show the maximum matching of graph $G$. Also, for a subset of vertices $V' \in V$, we let $G[V']$ be the induced subgraph of $G$ on vertices $V'$. Further, for subsets $U_1 \in V$ and $U_2 \in V$ such that $U_1 \cup U_2 = \emptyset$, we let $G[U_1, U_2]$ to show the induced bipartite subgraph between $U_1$ and $U_2$.

We say estimate $\widetilde{\mu}$ is a multiplicative $\alpha$-aproximation of $\mu(G)$ if $\alpha\cdot \mu(G)\leq \widetilde{\mu} \leq \mu(G)$. Also, We say estimate $\widetilde{\mu}$ is a multiplicative-additive $(\alpha, \beta n)$-approximation of $\mu(G)$ if
\begin{align*}
    \alpha \cdot \mu(G) - \beta n \leq \widetilde{\mu} \leq \mu(G).
\end{align*}

\paragraph{Problem Definition:} Given a graph $G$, we are interested in estimating the size of the maximum matching of $G$. We are given access to the adjacency list of the graph. In the adjacency list model, the list of neighbors of each vertex is stored in a list in an arbitrary order. The algorithm can query the $i$th neighbor of an arbitrary vertex $v$. The answer to the query is empty if vertex $v$ has less than $i$ neighbors.

\paragraph{Graph Theory:} We define a bipartite graph $H = (U,V,E)$ as {\em biregular} if the degree of all vertices in $U$ is identical, and likewise, the degree of all vertices in $V$ is identical. For a directed graph, we define its {\em underlying graph} as the undirected graph obtained by disregarding the direction of the edges.

\begin{definition}[Strongly Connected Component]
    Let $G$ be a directed graph. Then, $C$ is a strongly connected component of $G$ if it is a maximal set of vertices such that there exists a directed path between any pair of vertices $u$ and $v$ in $C$, and vice versa.
\end{definition}

We employ the well-known theorem by König \cite{Konig1916}, which states that the size of the minimum vertex cover equals the size of the maximum matching in bipartite graphs. Formally:

\begin{proposition}[König Theorem]
    The maximum matching size is equal to the minimum vertex cover size for any bipartite graph.
\end{proposition}

\paragraph{Probabilistic Tools:} The concentration inequalities utilized in this paper are as follows.

\begin{proposition}[Chernoff Bound]
    Let $X_1, X_2, \ldots, X_n$ be $n$ independent Bernoulli random variables. Let $X \sum_{i=1}^n X_i$. For any $k > 0$, it holds
    \begin{align*}
        \Pr[|X - \E[X]| \geq k] \leq 2 \exp \left(- \frac{k^2}{3\E[X]}\right).
    \end{align*}
\end{proposition}

\begin{definition}[Negative Association \cite{kumarDevProschen, saxenaKhursheed, wajc2017negative}]
    Let $X_1, X_2, \ldots, X_n$ be a set of random variables. We say this set is negatively associated if for any two disjoint index sets $I, J \subseteq [n]$, and two functions $f$ and $g$, both either monotonically increasing or monotonically decreasing, the following condition is satisfied:
    \begin{align*}
        \E[f(X_i: i \in I) \cdot g(X_j: j \in J)] \leq \E[f(X_i: i \in I)] \cdot \E[g(X_j: j \in J)].
    \end{align*}
\end{definition}

\begin{proposition}[Chernoff Bound Negatively Associated Variables]
    Let $X_1, X_2, \ldots, X_n$ be a set of negatively associated Bernoulli random variables. Let $X = \sum_{i=1}^n X_i$. Then,
    \begin{align*}
        \Pr\left[|X - \E[X]| \geq (1+\alpha) \E[X]\right] \leq \left( \frac{e^\alpha}{(1+\alpha)^{1+\alpha}} \right)^{\E[X]}.
    \end{align*}
\end{proposition}

\paragraph{Yao's Minimax Principle:} We use the following theorem to prove the lower bound for randomized algorithms.

\begin{proposition}[Yao's Minimax Principle \cite{Yao77}]\label{prop:yao-min-max}
    Suppose a problem is defined over the input $\mc{X}$. Let $\mc{A}$ be the set of all possible deterministic algorithms that solve this problem. Define $cost(a,x)$ to be the running time of algorithm $a\in \mc{A}$ on input $x\in \mc{X}$. Think of $p$ as a probability distribution over the selection of algorithms from $\mathcal{A}$, where $A$ stands for a randomly chosen algorithm based on $p$. Likewise, suppose $q$ is a probability distribution over the selection of inputs from $\mathcal{X}$, and $X$ is a representation of a randomly chosen input in accordance with $q$. It holds that:
    \begin{align*}
        \max_{x \in \mc{X}} \E[c(A, x)] \geq \min_{a \in \mc{A}} \E[c(a, X)].
    \end{align*}
\end{proposition}

\newpage
\section{Table of Parameters}\label{sec:tableofparameters}

In this section, we present a table of variables (\Cref{tbl:parameters}) employed in this paper. We assume that the algorithm makes $O(n^{2-\delta})$ queries. The table below provides definitions of these variables and their dependency on $\delta$. While there is no imperative need to read this section, we have already introduced these variables in the relevant sections. We include this table to facilitate readers' comprehension of the interplay between these parameters in the context of the technical proofs.

\setlength\extrarowheight{6pt}
\begin{center}
\begin{table}[h]
\begin{tabular}{ | >{\centering\arraybackslash} m{4.2em} | >{\centering\arraybackslash} m{7em} |  >{\centering\arraybackslash} m{29em}|}   
  \hline
  Parameter & Value & Definition\\ 
  \hline
   \hline
   $\delta$ & - & Parameter that controls the running time of the algorithm. More specifically, the algorithm has $O(n^{2-\delta})$ running time.\\
   \hline
   $L$ & $4/\delta$ & Number of \textbf{levels} in the recursive hierarchy for the construction of input distribution.\\
   \hline
   $r$ & $\left(\frac{10}{\delta}\right)^{L+1}$ & Number of \textbf{layers} in the base construction (and in each level of the hierarchy).\\
   \hline
   \yesdist & - & Distribution of graphs that have a perfect matching.\\
   \hline
   \nodist & - & Distribution of graphs that at most $(1-\epsilon)$ fraction of their vertices can be matched in the maximum matching.\\
   \hline
   $\yesdist^i$ & - & Distribution of level $i$ graphs in the construction hierarchy that have a perfect matching.\\
   \hline
   $\nodist^i$ & - & Distribution of level $i$ graphs that at most $(1-\epsilon)$ fraction of their vertices can be matched in the maximum matching.\\
   \hline
   $\dist$ & $\frac{1}{2}\yesdist + \frac{1}{2}\nodist$ & Final input distribution.\\
   \hline
    $\sigma_i$ & $\left( \frac{\delta}{10} \right)^{L+1-i}$ & Parameter that controls the degree of vertices in graphs of level $i$.\\
   \hline
   $d_i$ & $\Theta(n^{\sigma_i})$ & Parameter that controls the degree of vertices in graphs of level $i$.\\
   \hline
   $\zeta$ & $1/r^2$ & Fraction of vertices that are delusive in each level.\\
   \hline
   $\xi$ & $1/r^2$ & The gap between size of $A_r$ and $B_r$ in the base construction.\\
   \hline
   $\gamma$ & $1/r^3$ & Degree to delusive vertices is $\gamma d$.\\
   \hline
   $\tau$ & $(20r^3)^{-L}$ & Number of dummy vertices is $\tau n$.\\
   \hline
   $N_i$ & $N_i = n_{i - 1}/(2\zeta)$ &  Parameter that controls the number of vertices in graphs of level $i$.\\
   \hline
   $n_i$ & $(8 + 16r + 4\zeta r)N_i$ & Total number of vertices in a graph of level $i$.\\
   \hline
   $n$ & $(1+\tau)\cdot n_L$& Total number of vertices in a graph that is drawn from the final distribution.\\
\hline
\end{tabular}
\captionsetup{justification=centering}
\caption{Variables used in the input distribution and proofs.}
\label{tbl:parameters}
\end{table}
\end{center}

\clearpage

\section{Input Distribution and its Characteristics}\label{sec:input-distribution}

In this section, we describe the construction of our input distribution. We will have two types of input distributions both on $n$ vertices, which we denote by \yesdist{} and \nodist{}. Any graph drawn from $\yesdist$ will have a perfect matching which matches all $n$ vertices. On the flip side, any maximum matching for a graph drawn from $\nodist$ will match at most $(1-\epsilon)n$ vertices. Our final input distribution $\mc{D} := (\yesdist + \nodist)/2$ draws its graph either from $\yesdist$ or $\nodist$, each with probability 1/2. We show that any deterministic algorithm that can distinguish between a graph that is drawn from \yesdist{} and \nodist{}, has to spend at least $\Omega(n^{2-\delta})$ time. We fix the dependency of $\delta$ on $\epsilon$ later in the proofs. Our main result will be the following:

\begin{lemma}\label{lem:detlb}
    Let $G$ be drawn from $\mc{D}$. Any \textbf{deterministic} algorithm that provides an estimate $\widetilde{\mu}$ of the size of the maximum matching of $G$ such that
    $$
    \E_{G}[\widetilde{\mu}] \geq  \mu(G) - \epsilon n,
    $$
    will have to spend at least $\Omega(n^{2-\delta})$ time.
\end{lemma}

Plugging \cref{lem:detlb} into Yao’s minimax theorem \cite{Yao77}, we get our main result for randomized algorithms.

\begin{proof}[Proof of \Cref{thm:main}]
    Let $\mathcal{X}$ be the set of all possible inputs for the problem and $\mathcal{A}$ be the set of all possible deterministic algorithms. Also, let $c(a,x) \geq 0$ be the running time of the algorithm $a$ on input $x$. By \Cref{lem:detlb}, we have $\min_{a\in \mathcal{A}} \E[c(a, \mathcal{D})] \geq \Omega(n^{2-\delta})$. Therefore, using Yao's minimax principle (\Cref{prop:yao-min-max}), we have
    \begin{align*}
        \max_{x\in \mathcal{X}} \E[c(A,x)] \geq \min_{a\in \mathcal{A}} \E[c(a, \mathcal{D})] \geq \Omega(n^{2-\delta})
    \end{align*}
    which implies that any randomized algorithm that estimates the size of the maximum matching with an additive error of $\epsilon n$ must spend at least $\Omega(n^{2-\delta})$ time.
\end{proof}

For both \yesdist{} and \nodist{}, our construction consists of $L$ recursive {\em levels} of hierarchy. The level $i$ graph is constructed by combining several graphs of level $i-1$ plus extra edges to increase the difficulty in distinguishing the edges of the level $i-1$ graphs. The high-level goal is to hide some of the edges of (one of) the level 1 graphs in the construction, which consists of a constant fraction of the maximum matching edges of the graph. 

For each level $i$, there are two types of graphs which we call $\yesdist^i$ and $\nodist^i$. Similar to the \yesdist{} and \nodist{}, the two types of graphs for level $i$ have different sizes of maximum matching. Also, each level of the hierarchy consists of $r$ {\em layers}. In \Cref{subsec:base-construction}, we show how we construct our level 1 graph (base level of the hierarchy). Next, in \Cref{subsec:recursive}, we demonstrate how we can construct the core using a recursive process. Finally, in \Cref{subsec:dummy}, we add some dummy vertices which are a small constant fraction of vertices in the graph and we connect them to all vertices in order to increase the cost of adjacency list queries. Our \yesdist{} will be $\yesdist^{L}$ plus the dummy vertices and \nodist{} will be $\nodist^{L}$ plus the dummy vertices. It is also noteworthy to mention that all graphs in our constructions are bipartite.

\subsection{Base Level of the Hierarchy}\label{subsec:base-construction}

Let $N_1$ and $d_1$ be two parameters that control the number of vertices and degree of vertices in the induced subgraph of the base level.

\paragraph{Vertex set:} the vertex set of the base level consists of disjoint subsets of vertices $A_i^j$ and $B_i^j$ for $i \in [r]$ and $j \in \{1, 2\}$. Also, for each $i \in [r]$, the base level consists of subsets of vertices $D_i$ which we call {\em delusive} vertices. Finally, there are two subsets $S^1$ and $S^2$ in the construction. We have the following properties for the size of the subsets that we defined:
\begin{align*}
    |S^j| = |A_i^j| = |B_i^j| = N_1 \qquad  \forall i \in [r - 1] \quad \& \quad j \in \{1, 2\},
\end{align*}
\begin{align*}
    |B_r^1| = |B_r^2| = N_1, \qquad |A_r^1| = |A_r^2| = (1-\xi)N_1
\end{align*}
\begin{align*}
    |D_i| = \zeta N_1 \qquad \forall i \in [r]
\end{align*}
Let $n_1$ be the total number of vertices in the base-level construction. Thus,
\begin{align*}
    n_1 = |S^1| + |S^2| + \sum_{\substack{i \in [r] \\ j \in \{1,2\}}}|A_i^j| + |B_i^j| + \sum_{i \in [r]} |D_i| & = 2N_1+ 4rN_1 + \zeta r N_1\\
    &= (2 + 4r + 1/r)N_1 & (\text{Since } \zeta=1/r^2).
\end{align*}
Furthermore, we assume that all subsets have even size.

\paragraph{Edge set:} the edge set of $\yesdist^1$ and $\nodist^1$ are slightly different such that $\yesdist^1$ contains a perfect matching, however, a small fraction of vertices of graphs in $\nodist^1$ are unmatched in its maximum matching. The edge set consists of several biregular graphs between different subsets of vertices. Let $X$ and $Y$ be two different subsets of vertices. We use $\deg(X, Y)$ to show the degree of vertices of $X$ in the induced regular graph between $X$ and $Y$. In what follows, we determine the degree of vertices for different choices of $X$ and $Y$. We have the following biregular graphs in both $\yesdist^1$ and $\nodist^1$:
\begin{itemize}
    \item Edges of vertices in $S^j$ for $j\in \{1,2\}$:
    \begin{align*}
        \deg(S^j, B_1^j) = 1.
    \end{align*}
    \item Edges of vertices in $B_1^j$ for $j\in \{1,2\}$:
    \begin{align*}
        \deg(B_1^j, S^j) = 1, \qquad \deg(B_1^j, A_1^j) = d_1, \qquad \deg(B_1^j, D_1) = r\gamma d_1.
    \end{align*}
\item Edges of vertices in $A_i^j$ for $i\in[r-1]$ and $j\in \{1,2\}$:
    \begin{align*}
        \deg(A_i^j, B_i^j) = d_1, &\qquad \deg(A_i^j, B_{i+1}^j) = 1, \qquad \deg(A_i^j, D_i) = (r-i+1)\gamma d_1,\\
        &\deg(A_i^j, D_k) = \gamma d_1 \qquad \text{for $k < i$}.
    \end{align*}
\item Edges of vertices in $B_i^j$ for $1 < i \leq r$ and $j\in \{1,2\}$:
    \begin{align*}
        &\deg(B_i^j, A_i^j) = d_1, \qquad \deg(B_i^j, A_{i-1}^j) = 1, \qquad \deg(B_i^j, D_i) = (r-i+1)\gamma d_1,\\
        &\deg(B_i^j, D_k) = \gamma d_1 \qquad \text{for $k < i$}.
    \end{align*}
\item Edges of vertices in $D_r$:
    \begin{align*}
        &\deg(D_r, D_r) = d_1 + 1 + \gamma d_1 (1 - 4/\zeta + 2\xi/\zeta),\\
        &\deg(D_r , A_r^j) = (1-\xi)\gamma d_1/\zeta, \qquad \deg(D_r, B_r^j) = \gamma d_1/\zeta \qquad \text{for $j \in \{1, 2\}$},\\
        &\deg(D_r, D_i) = \gamma d_1 \qquad \text{for $i \in [r-1]$}.
    \end{align*}
\item Edges of vertices in $D_i$ for $i \in [r]$:
    \begin{align*}
        &\deg(D_i, D_i) = d_1 + 1 + \gamma d_1 - 2\gamma d_1(4r - 8i - \xi + 2)/\zeta,\\
        &\deg(D_i , A_i^j) = (r-i+1)\gamma d_1/\zeta, \qquad \deg(D_i, B_i^j) = (r-i+1)\gamma d_1/\zeta \qquad \text{for $j \in \{1, 2\}$},\\
        &\deg(D_i, D_k) = \gamma d_1 \qquad \text{for $k \neq i$},\\
        &\deg(D_i, A_k^j) = \gamma d_1/\zeta, \qquad \deg(D_i, B_k^j) = \gamma d_1/\zeta \qquad\text{for $i < k < r$ and $j \in \{1, 2\}$},\\
        &\deg(D_i, A^j_r) = (1-\xi) \gamma d_1/\zeta, \qquad \deg(D_i, B^j_r) = \gamma d_1/\zeta \qquad\text{$j \in \{1, 2\}$}.
    \end{align*}
\end{itemize}

Neighbors of vertices $A_r^j$ and $B_r^j$ for $j\in \{1, 2\}$ are slightly different in \yesdist{} and \nodist{}. In \yesdist{}, we add a random perfect matching between $A_r^1$ and $A_r^2$. Also, there exists a biregular graph between $A_r^j$ and $B_r^j$ such that the degree of vertices in $A_r^j$ is $d_1$ and the degree of vertices in $B_r^j$ is $(1-\xi)d_1$. Finally, we have a bipartite $(\xi d_1)$-regular graph between vertices of $B_r^1$ and $B_r^2$. Hence, the degrees are as follows in \yesdist{}:
\begin{itemize}
\item Edges of vertices in $B_r^j$:
    \begin{align*}
        &\deg(B_r^j, A_r^j) = (1-\xi)d_1, \qquad \deg(B_r^j, A_{r-1}^j) = 1, \qquad \deg(B_r^j, B_r^{3-j}) = \xi d_1, \\
        &\deg(B_r^j, D_k) = \gamma d_1 \qquad\text{$k \in [r]$}.
    \end{align*}
\item Edges of vertices in $A_r^j$:
    \begin{align*}
        &\deg(A_r^j, B_r^j) = d_1, \qquad \deg(A_r^j, A_{r}^{3-j}) = 1,\\
        &\deg(A_r^j, D_k) = \gamma d_1 \qquad\text{$k \in [r]$}.
    \end{align*}
\end{itemize}
In \nodist{}, we remove a $(1-\xi)N_1$ edges of a perfect matching of subgraph between $B_r^1$ and $B_r^2$. Let $\overline{B}_r^1$ and $\overline{B}_r^2$ be the set of vertices that are endpoints of the perfect matching in $B_r^1$ and $\overline{B}_r^2$, respectively. Also, we do not have a perfect matching between vertices of $A_r^1$ and $A_r^2$. Instead, we add a perfect matching between vertices of $A_r^j$ and $\overline{B}_r^j$ for $j \in \{1, 2\}$.

Note that for each of the regular bipartite subgraphs that we used in our construction, we choose one uniformly at random graph among all possible biregular graphs with specific degrees. Also, we assume that upon querying a vertex by the algorithm, if it belongs to $S^1$ or $S^2$, we immediately reveal the label of the vertex. What is hidden from the algorithm is whether the vertex belongs to subset $A$, $B$, or $D$ and the layer it belongs to. Now, we proceed to prove some characteristic properties of our base-level construction. The following observations are immediately implied by the construction.

\begin{remark}
To maintain the graphs bipartite, it is necessary to have two subsets within each $D_i$ because there are edges within each subset $D_i$. However, for the sake of simplicity in our construction, we omitted this aspect. It is possible to assume the presence of two subsets within each $D_i$ and add edges between these subsets to preserve the bipartite property of the graphs.
\end{remark}

\begin{observation}
    Let $v \in S^1 \cup S^2$. Then, the degree of $v$ in the base-level construction is 1.
\end{observation}

\begin{observation}
Let $v \notin S^1 \cup S^2$. Then, the degree of $v$ in the base level construction is $d_1 + \gamma d_1 + 1 = \Theta(d_1)$.
\end{observation}

Based on the two aforementioned observations, the degrees of all vertices are identical at the base level, except for vertices in $S^1 \cup S^2$. Additionally, as $\gamma$ is a small constant, we can assume that all vertices have approximately $O(d_1)$ neighbors at the base level. Next, we will demonstrate the contrast in the size of the maximum matching between a graph drawn from \yesdist{} and one drawn from \nodist{}.

\begin{lemma}\label{lem:matching-size-base}
Let $G_\yes^1 \sim \yesdist^1$ and $G_\no^1 \sim \nodist^1$. Then, we have
    $$
    \mu(G_\yes^1) = \frac{n_1}{2} \qquad \text{and} \qquad \mu(G_\no^1) \leq \frac{n_1}{2} - N_1.
    $$
\end{lemma}
\begin{proof}
    First, we prove that $G_\yes^1$ contains a perfect matching. There exists a perfect matching between the following subsets of vertices:
    \begin{itemize}
        \item $S^j$ and $B_1^j$ for each $j\in \{1,2 \}$,
        \item $A_i^j$ and $B^{i+1}_j$ for each $i\in [r-1]$ and $j\in \{1,2 \}$,
        \item $A_r^1$ and $A_r^2$,
        \item induced subgraph of $D_i$ for each $i\in[r]$ since the induced subgraph of vertices in $D_i$ is a bipartite regular graph.
    \end{itemize}
    Therefore, we have $\mu(G_\yes^1) = n_1/2$. On the other hand, for $G_\no^1$, combining one part of the bipartite graph of vertices in $D_i$ for all $i \in [r]$ and $\bigcup_{i=1, j\in \{ 1,2\}}^r B^j_i$ results in a vertex cover of the graph. Hence, using K\"{o}nig’s Theorem, we get 
    \begin{align*}
         \mu(G_\no^1)  \leq \sum_{i=1, j\in \{1,2\}}^r |B^j_i| + \sum_{i=1}^r |D_i|/2 = 2rN_1 + \frac{r\zeta N_1}{2} = (2r + \frac{1}{2r})N_1 = \frac{n_1}{2} - N_1. \qquad\qedhere
    \end{align*}
\end{proof}

\subsection{The Recursive Hierarchy}\label{subsec:recursive}
In this subsection, we show how we obtain our final construction from the base-level construction using a recursive procedure. We construct $\yesdist^\ell$ and $\nodist^\ell$ from $\yesdist^{\ell-1}$ and $\nodist^{\ell-1}$ for $1 < \ell \leq L$. Similar to the base level construction, each level has $r$ layers of vertices. Similarly, we have subsets $A_i^1$, $A_i^2$, $B_i^1$, and $B_i^2$ for $i\in[r]$. Moreover, we have two subsets $S^1$, $S^2$. However, instead of having one subset $D_i$ for $i \in [r]$, we have four subsets $D_i^j$ for $1\leq j \leq 4$. We let $D_i = \bigcup_{j=1}^4 D_i^j$. We let $A_i$ denote $A_i^1 \cup A_i^2$ (resp. $B_i$ denote $B_i^1 \cup B_i^2$ and $S$ denote $S^1 \cup S^2$). Henceforth, when we mention a vertex's membership in subset $X$ at level $\ell$ of the hierarchy, we are referring to $X$ one of the sets $A_i, B_i, D_i$, or $S$.

Let $N_\ell$ and $d_\ell$ be two parameters that control the number of vertices and degree of vertices in the
graph of level $\ell$. We have that $N_\ell = n_{\ell - 1}/(2\zeta)$. We have the following properties for the sizes of the subsets that we defined:
\begin{align*}
    |S^j| = |A_i^j| = |B_i^j| = 4N_\ell \qquad  \forall i \in [r - 1] \quad \& \quad j \in \{1, 2\},
\end{align*}
\begin{align*}
    |B_r^1| = |B_r^2| = 4N_\ell, \qquad |A_r^1| = |A_r^2| = 4N_\ell
\end{align*}
\begin{align*}
    |D_i^j| = \zeta N_\ell \qquad \forall i \in [r] \quad \& \quad 1 \leq j \leq 4,
\end{align*}
Let $n_\ell$ be the total number of vertices in level $\ell$ of construction. Thus,
\begin{align*}
    n_\ell = |S^1| + |S^2| + \sum_{\substack{i \in [r] \\ j \in \{1,2\}}}|A_i^j| + |B_i^j| + \sum_{i \in [r]} |D_i| &= (8 + 16r + 4\zeta r)N_\ell\\
    & = (4/\zeta + 8r/\zeta + 2r) n_{\ell - 1} & (\text{Since } N_\ell = n_{\ell - 1}/(2\zeta)).
\end{align*}
We can also write the number of vertices in level $\ell$ in terms of $N_1$, which is the parameter that controls the number of vertices in the base level.

\begin{observation}\label{obs:vertices-count-level-l}
    It holds that $n_\ell = (2 + 4r + \zeta r)\cdot(4/\zeta + 8r/\zeta + 2r)^{\ell - 1}\cdot N_1$.
\end{observation}

\begin{observation}\label{obs:final-bound-number-vertices}
    $n_L \leq (9r^3)^L \cdot N_1$.
\end{observation}
\begin{proof}
    By \Cref{obs:vertices-count-level-l}, we have
    \begin{align*}
         n_L & = (2 + 4r + \zeta r)\cdot(4/\zeta + 8r/\zeta + 2r)^{L - 1}\cdot N_1\\
       & \leq (4r^2 + 8r^3 + 2r)^{L}\cdot N_1 & (\text{Since } \zeta = 1/r^2)\\
       & \leq (9r^3)^L \cdot N_1 & (\text{Since } r\geq10). \qquad  \qedhere
    \end{align*}
\end{proof}

Furthermore, we assume that all subsets have even size. The following edges are common in both $\yesdist^\ell$ and $\nodist^\ell$:
\begin{itemize}
    \item For $j \in \{1, 2\}$, there are $4/\zeta$ bipartite graphs that are drawn from $\yesdist^{\ell - 1}$ with disjoint vertex sets between $S^j$ and $B_1^j$.
    \item For $j \in \{1, 2\}$ and $i\in [r-1]$, there are $4/\zeta$ bipartite graphs that are drawn from $\yesdist^{\ell - 1}$ with disjoint vertex sets between $B_i^j$ and $A_{i+1}^j$.
    \item For $i \in [r]$ and $j \in \{1, 3\}$, there exists a bipartite graph that is drawn from $\yesdist^{\ell - 1}$ between $D_i^j$ and $D_i^{j+1}$.
\end{itemize}
Also, the edge set contains several biregular graphs similar to the construction of the base level. In what follows, we determine the degree of vertices for different choices of $X$ and $Y$ using the same notation of $\deg(X, Y)$.
\begin{itemize}
    \item Edges of vertices in $A_i^j$ for $i\in [r]$ and $j\in \{1,2\}$:
    \begin{align*}
        \deg(A_i^j, B_i^j)& = d_\ell, \qquad \deg(A_i^j, D_i)=(r-i+1)\gamma d_\ell,\\
        &\deg(A_i^j, D_k) = \gamma d_\ell \qquad \text{for $k < i$}.
    \end{align*}
    \item Edges of vertices in $B_i^j$ for $i\in [r]$ and $j\in \{1,2\}$:
    \begin{align*}
        \deg(B_i^j, A_i^j) &= d_\ell, \qquad \deg(B_i^j, D_i)=(r-i+1)\gamma d_\ell,\\
        &\deg(B_i^j, D_k) = \gamma d_\ell \qquad \text{for $k < i$}.
    \end{align*}
    \item Edges of vertices in $D_i$ for $i \in [r]$:
    \begin{align*}
        &\deg(D_i, A_i^j) = (r-i+1)\gamma d_\ell/\zeta, \qquad \deg(D_i, B_i^j)=(r-i+1)\gamma d_\ell/\zeta \qquad \text{for $j\in \{1,2\}$},\\
        &\deg(D_i, A_k^j) = \gamma d_\ell/\zeta, \qquad \deg(D_i, B_k^j)=\gamma d_\ell/\zeta \qquad \text{for $k > i$ and $j\in \{1,2\}$}.
    \end{align*}
\end{itemize} 
Further, for $i \in [r]$ and $j \in \{1, 2\}$, there exists a biregular graph between $D^j_i$ and $D^{j+2}_i$ with degree $d_\ell + \gamma d_\ell  - 4\gamma d_\ell(2r - 2i + 1)/\zeta$. Also, since we have four parts in each $D_i$, we can add edges between other vertices and corresponding subsets in $D_i$ to keep the graph bipartite. For simplicity, we skip the detailed degrees of this part since it is only important to keep the graph bipartite and the reader can assume that we have a set $D_i$ and ignore about how edges are inside the set.

The only difference between $\yesdist^\ell$ and $\nodist^\ell$ is the subgraph between $A_r^1$ and $A_r^2$. In $\yesdist^\ell$, this subgraph is drawn from $\yesdist^{\ell-1}$ and in $\nodist^\ell$, this subgraph is drawn from $\nodist^{\ell-1}$. The following observations are immediately implied by the
construction.

\begin{observation}
    Let $G$ be a graph that is drawn from $\yesdist^{\ell}$ or $\nodist^{\ell}$. Suppose that we remove all subgraphs that are drawn from $\yesdist^{\ell-1}$ and $\nodist^{\ell-1}$ during the recursive construction of $G$. Then, the degree of each vertex in $S^1 \cup S^2$ is 0. Moreover, the degree of vertices that are not in $S^1 \cup S^2$ is $d_\ell + \gamma d_\ell$.
\end{observation}

\begin{observation}
    Degree of vertices in a graph that is drawn from $\yesdist^{\ell}$ or $\nodist^{\ell}$ is $O(d_\ell)$.
\end{observation}

\begin{observation}\label{obs:unique-level-edge}
    For every pair of vertices $u$ and $v$, there is a unique level $\ell$ such that if there is an edge between them at all, it must belong to level $\ell$.
\end{observation}

\begin{lemma}\label{lem:matching-size-induc}
    Let $G_\yes^\ell \sim \yesdist^\ell$ and $G_\no^\ell \sim \nodist^\ell$. Then, we have
    \begin{itemize}
        \item $\mu(G_\yes^\ell) = n_\ell/2$,
        \item $\mu(G_\no^\ell) \leq n_\ell/2 - N_1$.
    \end{itemize}
\end{lemma}

\begin{proof}
    We use induction to prove this lemma. For the base case where $\ell = 1$, the proof follows by \Cref{lem:matching-size-base}. Similar to the proof of \Cref{lem:matching-size-base}, we can show that $G_\yes^\ell$ has a perfect matching since there exists a perfect matching between the following subsets of vertices:
    \begin{itemize}
        \item $S^j$ and $B_1^j$ for each $j\in \{1,2 \}$,
        \item $A_i^j$ and $B^{j}_{i+1}$ for each $i\in [r-1]$ and $j\in \{1,2 \}$,
        \item $A_r^1$ and $A_r^2$,
        \item $D^1_i$ and $D^2_i$ for all $i\in[r]$,
        \item $D^3_i$ and $D^4_i$ for all $i\in[r]$.
    \end{itemize}
    All the above subgraphs are vertex disjoint and have a perfect matching because their subgraph is drawn from $\yesdist^{\ell-1}$. Therefore, we have $\mu(G_\yes^\ell) = n_\ell/ 2$.

    For $G_\no^\ell$, note that if we remove edges between $A_r^1$ and $A_r^2$, then the size of maximum matching is at most 
    \begin{align}\label{eq:first-match-size}
        \sum_{i=1, j\in \{1,2\}}^r |B^j_i| + \sum_{i=1}^r |D_i|/2 = 8rN_\ell + 2r\zeta N_\ell = (8r + \frac{2}{r})N_\ell,
    \end{align}
    since combining one part of the bipartite graph of vertices in $D_i$ for all $i \in [r]$ and $\bigcup_{i=1, j\in \{ 1,2\}}^r B^j_i$ results in a vertex cover of the graph. Also, by the induction hypothesis, we have
    \begin{align}\label{eq:second-match-size}
        \mu(G_\no^\ell[A_r^1, A_r^2]) \leq \frac{4}{\zeta}\cdot \mu(G_\no^{\ell-1}) \leq \frac{4}{\zeta} \cdot \left(\frac{n_{\ell - 1}}{2} - N_1\right) \leq 4N_\ell- N_1.
    \end{align}
    Summing up \Cref{eq:first-match-size} and \Cref{eq:second-match-size}, we obtain
    \begin{align*}
        \mu(G_\no^\ell) \leq (8r + \frac{2}{r} + 4)N_\ell - N_1 = \frac{n_\ell}{2} - N_1. \qquad \qedhere
    \end{align*}
\end{proof}

\subsection{Adding Dummy Vertices}\label{subsec:dummy}

Finally, in both \yesdist{} and \nodist{}, we add $\tau n_L$ dummy vertices to the whole graph and connect these vertices to all other vertices in the graph. Also, we assume that $\tau n_L$ is an even number and we keep the graph bipartite after adding $\tau n_L$ vertices, i.e. half of the dummy vertices are connected to one part of the graph, and the other half are connected to the other part. Further, we assume that there is a perfect matching between dummy vertices in order to have a perfect matching in \yesdist{}. The intuition behind adding dummy vertices to the graphs in our input distribution is that they will increase the cost of adjacency list queries while the size of the matching does not change that much since $\tau$ is a very small constant. Moreover, we assume that the algorithm knows which vertices are dummy. We use {\em core} to denote the induced subgraph of all vertices excluding dummy vertices.

\begin{observation}\label{obs:tau-bound}
    $\tau n_L \leq N_1 / 2$.
\end{observation}
\begin{proof}
By \Cref{obs:final-bound-number-vertices}, we have
\begin{align*}
       \tau n_L & \leq \tau \cdot (9r^3)^L\cdot N_1\\
       & = (20r^3)^{-L} \cdot (9r^3)^L \cdot N_1 & (\text{Because of our choice of }\tau)\\
       & \leq N_1/2 \qquad \qedhere
   \end{align*}
\end{proof}

\begin{claim}\label{clm:final-matching-size-yes-no}
    Let $G_\yes \sim \yesdist$ and $G_\no \sim \nodist$. Then, we have
    \begin{itemize}
        \item $\mu(G_\yes) = n_L\cdot(1 + \tau) / 2$,
        \item $\mu(G_\no) \leq n_L / 2 - N_1 / 2$.
    \end{itemize}
\end{claim}

\begin{proof}
    Combining \Cref{lem:matching-size-induc} and the fact that there exists a perfect matching in the induced subgraph of dummy vertices implies that $G_\yes$ has a perfect matching. Thus, $\mu(G_\yes) = n_L\cdot(1 + \tau) / 2$.

    If we remove dummy vertices, the size of the maximum matching in $G_\no$ is at most $n_L/2 - N_1$ by \Cref{lem:matching-size-induc}. On the other hand, there are at most $\tau n_L$ edges in the maximum matching of $G_\no$ with at least one dummy endpoint. Hence,
    \begin{align*}
        \mu(G_\no) \leq n_L/2 - N_1 + \tau n_L \leq n_L/2 - N_1/2,
    \end{align*}
    where the last inequality follows by \Cref{obs:tau-bound}.
\end{proof}

\begin{lemma}\label{lem:matching-size-final-graph}
    Let $\epsilon = (\delta/400)^{100/\delta^2}$. Any algorithm that estimates the size of the maximum matching of a graph that is drawn from the input distribution with $\epsilon n$ additive error must be able to distinguish whether it belongs to \yesdist{} or \nodist{}.
\end{lemma}

\begin{proof}
    Let $G_\yes \sim \yesdist$ and $G_\no \sim \nodist$. By \Cref{clm:final-matching-size-yes-no}, we have
    \begin{align*}
       \mu(G_\yes) - \mu(G_\no) \geq n_L\cdot(1 + \tau) / 2 - n_L/2 + N_1/2 \geq N_1/2.
    \end{align*}
    So it is enough to show that $\epsilon n \leq N_1/2$. To see this
    \begin{align*}
        \epsilon n = \epsilon n_L (1 + \tau) &\leq 2\epsilon n_L\\
        & \leq 2\epsilon \cdot (9r)^{3L} \cdot N_1 & (\text{By \Cref{obs:final-bound-number-vertices}})\\
        & \leq 2\epsilon \cdot (90/\delta)^{100/\delta^2} \cdot N_1 & (\text{Because of our choices of }r\text{ and }L) \\
        & \leq N_1 / 2 & (\text{Since } \epsilon = (\delta/400)^{100/\delta^2}). \qquad\qedhere
    \end{align*}
\end{proof}

Furthermore, we want to stress that the adjacency list of each vertex includes its neighbors in a random order. This ordering is chosen uniformly and independently for each vertex. In the rest of the paper, we assume that $d_1 = n^{\sigma_1}, \ldots, d_L = n^{\sigma_\ell}$ where $\sigma_L \gg \sigma_{L - 1} \gg \ldots \gg \sigma_1$. More specifically, we have 
\begin{align*}
    \sigma_i = \left( \frac{\delta}{10} \right)^{L+1-i},
\end{align*}
for $i \in [L]$. Also, we let $\sigma_{L+1} = 1$ and $\sigma_0=0$.

\section{Indistinguishability of the \yes{} and \no{} distributions}\label{sec:indistinguishability}

In this section, we show that an algorithm that makes at most $Q = O(n^{2-\delta})$ adjacency list queries, cannot distinguish if a graph is drawn from \yesdist{} or \nodist{}. Note that when an algorithm makes $O(n^{2-\delta})$ queries, it might see some cycles in the queried subgraph of the core (ignoring edges to dummy vertices that we added to increase the cost of adjacency list queries). In contrast, all the previous lower bounds for sublinear matching use the fact that the queried subgraph is a forest and the same approach cannot extend to get stronger lower bounds. We show that although the algorithm discovers cycles in the graph that is drawn from our input distribution, these cycles cannot be useful in distinguishing essential edges that are different in \yesdist{} and \nodist{}.

\subsection{Warm-Up: The algorithm cannot identify many edges that do not belong to the top level}\label{sec:warm-up}

It is important to keep in mind that the difference between a graph that is drawn from \yesdist{} and a graph that is drawn from \nodist{} stems from the subgraph between $A_r^1$ and $A_r^2$ of the highest level. In \yesdist{}, this subgraph is drawn from $\yesdist^{L-1}$ and in \nodist{}, this subgraph is drawn from $\nodist^{L-1}$. Thus, any algorithm that distinguishes between \yesdist{} and \nodist{}, should find the difference in this subgraph. In this subsection, we provide an upper bound on the number of edges that the algorithm can identify as belonging to this subgraph. In the following definition, we establish the notion of identifying or distinguishing an edge that belongs to the subgraph $A^1_r$ and $A^2_r$ in the following definition. When the algorithm queries a typical edge, because of our choices of $d_L$ and $d_{L-1}$, we expect the probability that this edge belongs to subgraph $A^1_r$ and $A^2_r$ to be roughly equal to $d_{L-1} / d_{L} = n^{\sigma_{L-1} - \sigma_L}$. We say an edge can be identified when the algorithm has a bias on this probability condition on the subgraph that is queried by the algorithm.

\begin{definition}[$p^{inner}_e$ and distinguishability of an edge]\label{def:distinguishing-definition}
    Let $e$ be an edge that is queried by the algorithm. Also, let $p^{inner}_e$ be the probability that this edge belongs to the subgraph between $A_r^1$ and $A_r^2$ conditioned on all queries made by the algorithm so far and assuming either input distribution. We say the algorithm can {\em distinguish} or {\em identify} if $e$ belongs to the subgraph between $A_r^1$ and $A_r^2$ if $p^{inner}_e > 10n^{\sigma_{L-1}-\sigma_L}$.
\end{definition}

Note that each vertex is adjacent to $O(d_L)$ edges in our core by our choice of $d_1, d_2, \ldots, d_L$. Further, each vertex is adjacent to $\Omega(n)$ dummy vertices that we added to the construction in order to increase the cost of adjacency list queries inside the core. Since the adjacency list of each vertex is ordered uniformly at random, each query to the adjacency list of a vertex results in an edge in the core with probability $O(d_L/n)$. Hence, we expect to have $O(d_L \cdot n^{1-\delta}) = O(n^{1-\delta + \sigma_L})$ queries inside the core since there are at most $O(n^{2-\delta})$ queries in total. We prove that the number of edges that the algorithm can identify as belonging to the subgraph between $A_r^1$ and $A_r^2$ is upper bounded by $O(n^{1-2\delta + 4\sigma_L})$. Moreover, we show that for all other edges, the probability that the edge belongs to the subgraph between $A_r^1$ and $A_r^2$ is $O(n^{\sigma_{L-1} - \sigma_L})$. 

In the next claim, we give an upper bound on the total number of edges without a dummy endpoint that the algorithm can query.

\begin{claim}\label{clm:discovered-edges-bound}
    Any algorithm that makes at most $Q = O(n^{2-\delta})$ queries, identifies at most $O(n^{1-\delta +\sigma_L})$ edges of the core with high probability.
\end{claim}
\begin{proof}
There are at most $O(n^{1-\delta})$ vertices such that the algorithm makes more than $\tau n_L / 2$ adjacency list queries to them since the total number of queries is $O(n^{2-\delta})$. For each vertex that the algorithm makes more than $\tau n_L / 2$ queries, we assume that the algorithm finds all its incident edges in the core which is at most $O(d_L) = O(n^{\sigma_L})$ and in total is at most $O(n^{1-\delta + \sigma_L})$. 

Now consider a vertex $v$ with at most $\tau n_L/2$ adjacency list queries. At the beginning of the algorithm, each query to $v$'s adjacency list is in the core with probability at most $O(n^{\sigma_L}/n)$. While the algorithm has made at most $\tau n_L /2$ queries, the queries made have only negligible effect on this probability, so it remains true that each query to $v$'s adjacency list is in the core with probability at most $O(n^{\sigma_L}/n)$. Let $X_i$ be the event that the $i$th query returns an edge in the core and let $X = \sum X_i$. Thus, $\E[X_i] \leq O(n^{\sigma_L - 1})$ and $\E[X] \leq O(Q\cdot n^{\sigma_L - 1}) = O(n^{1-\delta + \sigma_L})$. Further, random variables $X_i$s are negatively correlated. Therefore, using Chernoff bound we have
\begin{align*}
    \Pr\left[|X - \E[X]| \geq 6\sqrt{\E[X] \log n}\right] \leq 2\exp \left(-\frac{(6\sqrt{\E[X] \log n})^2}{3\E[X]} \right) \leq \frac{1}{n^{10}},
\end{align*}
which implies that with a probability of at least $1-n^{-10}$, the total number of edges in the core that is discovered by the algorithm is $O(n^{1-\delta +\sigma_L})$.
\end{proof}

Let us consider a scenario where, instead of the bipartite subgraphs found in our input distribution, we had Erdos-Renyi subgraphs with the same expected degree as the regular graphs. In this case, for a pair of vertices between which the algorithm has not yet discovered an edge, the probability of an edge's existence was upper-bounded by $O(d/n)$, where $d$ represents the expected degree of vertices in that subgraph. We extend this observation and employ a coupling argument to establish a similar property, which is formally articulated in the subsequent lemma, for the graphs generated from our input distribution.

\begin{lemma}\label{lem:edge-prob-bound-remaining}
    Let $(u, v)$ be a pair of vertices in the core that is not among discovered edges by the algorithm. Consider a time during the execution of the algorithm that $u$ and $v$ have $x$ and $y$ undiscovered core edges, respectively; suppose further that $x \leq y$. Then, there exists edge $(u,v)$ in the graph with probability at most $O(x/n)$.
\end{lemma}
\begin{proof}
    First, if $x = 0$ the edge exists with probability zero. Now, suppose that $x > 0$. According to the construction, if there exists an edge between $u$ and $v$, it only exists in one level of the recursive construction by \Cref{obs:unique-level-edge}. Let $X_u$ and $X_v$ be the subsets in the construction that $u$ and $v$ belong to in that level. If there are no edges between $X_u$ and $X_v$, then the probability of having edge $(u, v)$ is zero. Let $d_u$ be the number of neighbors of $u$ in $X_v$ and $d_v$ be the number of neighbors of $v$ in $X_u$. Also, let $\mathcal{G}_{(u,v)}$ be the set of all graphs in our input distribution that have all discovered edges in the core and edge $(u,v)$. On the other hand let $\overline{\mathcal{G}_{(u,v)}}$ be the set of all graphs in our input distribution that have all discovered edges in the core and do not have edge $(u,v)$. We prove that $|\mathcal{G}_{(u,v)}| / |\overline{\mathcal{G}_{(u,v)}}| = O(x/n)$ which implies that the probability of existence of edge $(u,v)$ is upper bounded by $O(x/n)$.

    To show this claim holds, for each graph $G_{(u,v)} \in \mathcal{G}_{(u,v)}$ we find all pairs $(u', v')$ such that $u' \in X_u$, $v' \in X_v$, the induced subgraph of $\{u,u',v,v'\}$ exactly has two edge $(u,v)$ and $(u', v')$, and edge $(u', v')$ has not been discovered by the algorithm. Then, by removing edges $(u,v)$ and $(u',v')$, and replacing them with edges $(u, v')$ and $(u', v)$ we get a graph in our input distribution that is in $\overline{\mathcal{G}_{(u,v)}}$. 
    
    We now argue that there are many such $(u',v')$ pairs. First, recall that $|X_u| = \Omega(n)$, $|X_v| = \Omega(n)$, $d_u \ll n$, and $d_v \ll n$. Thus most vertices in $X_v$ are not adjacent to $u$; in particular, $|X_v \setminus \mathcal{N}(u)| = \Omega(n) $. Let $P$ be the set of all edges $(w,z)$ such that $w \in X_v \setminus \mathcal{N}(u)$, $z \in X_u$ ($(w,z)$ may or may not have been discovered by the algorithm). Since each vertex in $X_v$ has $d_v$ neighbors in $X_u$ and $|X_v \setminus \mathcal{N}(u)| = \Omega(n)$, we get $|P| = \Omega(nd_v)$. Now let $P'$ be the subset of edges in $P$ that have not been discovered by the algorithm. By \Cref{clm:discovered-edges-bound}, the total number of discovered edges by the algorithm is $o(n)$ which implies that $|P'| = \Omega(nd_v)$. It is not hard to see each pair $(u', v') \in P'$ satisfies all the required conditions. 
    
    Hence, we can map each graph in $\mathcal{G}_{(u,v)}$ to at least $\Omega(nd_v)$ graphs in $\overline{\mathcal{G}_{(u,v)}}$. Conversely, in the case of each graph in $\overline{\mathcal{G}_{(u,v)}}$, it can be mapped to a maximum of $xy$ graphs in $\mathcal{G}_{(u,v)}$, considering that the remaining undiscovered edges of $u$ and $v$ are $x$ and $y$, respectively. Therefore, by double counting the edge of the mapping from both sides, it holds
    \begin{align*}
        \frac{|\mathcal{G}_{(u,v)}|}{|\overline{\mathcal{G}_{(u,v)}}|} \leq \frac{O(xy)}{\Omega(nd_v)} \leq \frac{O(xd_v)}{\Omega(nd_v)} \leq O\left(\frac{x}{n}\right),
    \end{align*}
    which completes the proof. Furthermore, it is crucial to mention that in this mapping, every graph in the support of \yesdist{} is mapped with graphs solely in the support of \yesdist{}, and likewise, every graph in the support of \nodist{} is mapped with graphs solely from the support of \nodist{} since both $u$ and $u'$ belong to the same subset in the construction, and similarly, $v$ and $v'$ belong to the same subset in the construction as well.
\end{proof}

\begin{corollary}\label{cor:edge-prob-bound}
   At any point during the execution of the algorithm, for any pair of vertices $(u, v)$ in the core that is not among the edges already discovered by the algorithm, there $(u,v)$ is an edge in the graph with probability at most $O(n^{\sigma_L - 1})$.
\end{corollary}
\begin{proof}
    At any point during the execution of the algorithm, there are $O(n^{\sigma_L})$ undiscovered edges in the core incident on $u$ or $v$. Plugging this into \Cref{lem:edge-prob-bound-remaining} we obtain the claimed bound.
\end{proof}

\begin{definition}[Direction of an Edge]
    Let $(u,v)$ be an edge that is queried by the algorithm by making a query to the adjacency list of vertex $u$. When we refer to the direction of edge $(u, v)$, we are indicating that it goes from $u$ to $v$.
\end{definition}

In the next claim, we show that for any fixed pair $(u,v)$, when the algorithm queries $u$'s adjacency list the answer is $v$ with probability at most $O(1/n)$, even when conditioning on the query returning a non-dummy vertex.

\begin{claim}\label{clm:incoming-prob}
    Suppose that the algorithm queries the adjacency list of vertex $u$ in the core. Let $v$ be a vertex in the core that the algorithm has not discovered edge $(u,v)$ yet. Then, the probability of getting $v$ as the answer to the adjacency list query of vertex $u$ is at most $O(1/n)$.
\end{claim}
\begin{proof}
    Suppose that there are $x$ remaining undiscovered edges of $u$ at the time that the algorithm is making a query to the adjacency list of $u$. By \Cref{lem:edge-prob-bound-remaining}, the probability of having an edge between $u$ and $v$ is $O(x/n)$. Now assume that there exists an edge $(u,v)$. Since $u$ has $x$ undiscovered edges and the adjacency list of vertices is sorted in a random order, the probability of $v$ being the first one is $1/x$ condition on the edge existence. Therefore, the probability of getting $v$ as the answer to the adjacency list query of vertex $u$ is at most $O(1/n)$.
\end{proof}

As an application of \Cref{clm:incoming-prob}, we can demonstrate that each vertex in the graph has an indegree of  $O(\log n)$ because they all have a nearly uniform probability of being the answer to the adjacency list queries.

\begin{claim}\label{clm:max-in-degree}
    With high probability, the indegree of every vertex is at most $5 \log n$.
\end{claim}
\begin{proof}
    Let $k$ be the number of edges that the algorithm finds in the core. By \Cref{clm:discovered-edges-bound}, we have $k = O(n^{1-\delta + \sigma_L})$. Consider an arbitrary vertex $v$. For $i \in [k]$, let $X_i$ be the event that $i$th queried edge in the core be an incoming edge to $v$. By \Cref{clm:incoming-prob}, we have that $\Pr[X_i = 1] = O(1/n)$ for all $i$. Let $X = \sum_{i=1}^k X_i$ and $\lambda = (4\log n)/\E[X]$. Also, $0 < \E[X] < 1$ and thus, $\lambda > e^2$ for large enough $n$. Note that $X_i$'s are negatively associated random variables. Using Chernoff bound for negatively associated variables, we have
    \begin{align*}
        \Pr\left[X \geq (1+\lambda)\E[X] \right] &\leq \left(\frac{e^\lambda}{(1 + \lambda)^{1+\lambda}}\right)^{\E[X]}\\
        & \leq \left( \frac{e^\lambda}{\lambda^\lambda} \right)^{\E[X]} & (\text{Since } \lambda > 1)\\
        & = \left(\frac{e}{\lambda}\right)^{4\log n} &(\text{Since } \lambda = (4\log n)/\E[X])\\
        & \leq \frac{1}{n^4} &(\text{Since } \lambda > e^2)
    \end{align*}
    Since $(1+\lambda)\E[X] = \E[X] + 4\log n < 5\log n$, the probability that $v$ has more than $5\log n$ incoming edges is at most $n^{-4}$. Using a union bound over all vertices we get the claimed bound.
\end{proof}

\begin{definition}[Shallow Subgraph]\label{def:shallow-subgraph}
For a vertex $v$, we let $v$'s {\em shallow subgraph}  be the set of vertices that are reachable from $v$ using queried subgraph directed paths of length at most $10\log n$. We use $T(v)$ to denote $v$'s shallow subgraph. 
\end{definition}

We can utilize \Cref{clm:incoming-prob} to establish a more robust proposition than what \Cref{clm:max-in-degree} offers. To clarify, we can demonstrate that the algorithm is unable to concentrate outgoing edges towards nearby vertices. Consequently, the majority of vertices that are close together in the queried subgraph will have only one incoming edge. As a result, each vertex will be part of $\widetilde{O}(1)$ shallow subgraphs.

\begin{lemma}\label{cor:vertex-in-shallow-count}
    With high probability, each vertex is in at most $\widetilde{O}(1)$ shallow subgraphs.
\end{lemma}
\begin{proof}
    Let $v$ be an arbitrary vertex in the core. Suppose that we run a BFS from $u$ in the queried subgraph with reverse edge directions and let $V_i$ be the set of vertices that are in distance $i$ from $v$ for $i \in [10 \log n]$. We show that with high probability, we have $|V_i| \leq  i \log^2 n$. We do this using induction. For $i=1$, the claim is held by \Cref{clm:max-in-degree}. Suppose that the claim holds for all $i'$ such that $i' < i$. Let $u \in V_{i-1}$. By \Cref{clm:incoming-prob}, the probability that the algorithm makes a query that is an incoming edge to $u$ is at most $O(1/n) \leq \log n /n$ for a large enough $n$. Also, we have that $|V_{i-1}| \leq (i-1)\cdot\log^2 n$. Hence, the probability that a queried edge goes to one of the vertices in $V_{i-1}$ is at most $(i-1)\cdot \log^3 n /n$. Let $k$ be the total number of edges the algorithm finds in the core. By \Cref{clm:discovered-edges-bound}, we have $k\leq n^{1-\delta + \sigma_L} \cdot \log n$.

    For $i \in [k]$, let $X_i$ be the event that $i$th queried edge in the core be an incoming edge to $V_{i-1}$. Thus, we have that $\Pr[X_i = 1] \leq (i-1)\cdot \log^3 n /n$ for all $i$. Let $X = \sum_{i=1}^k X_i$ and $\lambda = (4\log n)/\E[X]$. Hence, $\E[X] \leq (i-1)\cdot \log^4 n \cdot n^{\sigma_L - \delta}$. Also, $0 < \E[X] < 1$ and thus, $\lambda > e^2$ for large enough $n$. Note that $X_i$'s are negatively associated random variables. Using Chernoff bound for negatively associated variables, we have
    \begin{align*}
        \Pr\left[X \geq (1+\lambda)\E[X] \right] &\leq \left(\frac{e^\lambda}{(1 + \lambda)^{1+\lambda}}\right)^{\E[X]}\\
        & \leq \left( \frac{e^\lambda}{\lambda^\lambda} \right)^{\E[X]} & (\text{Since } \lambda > 1)\\
        & = \left(\frac{e}{\lambda}\right)^{4\log n} &(\text{Since } \lambda = (4\log n)/\E[X])\\
        & \leq \frac{1}{n^4} &(\text{Since } \lambda > e^2)
    \end{align*}
    which implies that $|V_i| \leq i \log^2 n$ since $(1+\lambda)\E[X] = \E[X] + 4\log n < i\log^2 n$. Therefore,
    \begin{align*}
        \sum_{i=1}^{10\log n} |V_i| \leq \sum_{i=1}^{10\log n} i \cdot \log^2 n \leq \widetilde{O}(1). \qquad \qedhere
    \end{align*}
\end{proof}

\begin{corollary}\label{lem:edge-in-shallow-count}
    With high probability, each edge that the algorithm finds in the core is in at most $\widetilde{O}(1)$ shallow subgraphs.
\end{corollary}
\begin{proof}
    For edge $(u,v)$, by \Cref{cor:vertex-in-shallow-count}, $u$ is in at most $\widetilde{O}(1)$ shallow subgraphs. Therefore, edge $(u,v)$ is in at most $\widetilde{O}(1)$ shallow subgraphs.
\end{proof}

\paragraph{Spoiled vertices:} In essence, spoiled vertices are those in close proximity to short cycles using directed edges or having large shallow subgraphs. Later, we can prove that, for vertices distanced from short cycles or lacking large shallow subgraphs, the algorithm cannot distinguish if their incoming edges originate from the inner hierarchy level.

Before we formally define spoiled vertices, we define a closely related notion of {\em spoiler vertices}.
Intuitively, spoiler vertices are ones where the idealized forest structure of the queried core subgraph is violated (or ``spoiled'').

\begin{definition}[Spoiler Vertex]\label{def:spoiler-vertex}
   We say a vertex $u$ in the core is {\em spoiler} if at least one of the following conditions holds:
    \begin{itemize}
        \item[(i)] vertex $u$ has more than one incoming edge, 
        \item[(ii)] there is an edge $(u,v)$ that is discovered by the algorithm at a time when $v$ already has non-zero degree.
    \end{itemize}
\end{definition}

Spoiled vertices are ones that have, or expect to have, spoiler vertices in their shallow subgraphs.

\begin{definition}[Spoiled Vertex]\label{def:spoiled-vertex}
    A vertex $v$ in core is {\em spoiled} if its shallow subgraph contains any of the following:
    \begin{itemize}
        \item a spoiler vertex; or
        \item at least $n^{\delta - 2\sigma_L}$ vertices.
    \end{itemize}
\end{definition}

The following observation is directly implied by the way we defined spoiler and spoiled vertices.

\begin{observation}\label{obs:tree-structure-of-unspoiled}
    Let $v$ be a vertex that is not spoiled. Then, the shallow subgraph of $v$ is a rooted tree of size at most $n^{\delta - 2\sigma_L}$. Moreover, for each edge $(u,w)$ in the shallow subgraph of $v$, at the time that the algorithm made the query, $w$ was a singleton vertex.
\end{observation}

\begin{proof}
    At the time that the algorithm discovers an edge $(u,w)$ in the shallow subgraph of $v$, vertex $w$ should be singleton according to \Cref{def:spoiler-vertex} and \Cref{def:spoiled-vertex}. Therefore, the shallow subgraph of $v$ is a rooted tree.
\end{proof}

In the next lemma, we show that even among vertices for which the algorithm finds a core edge, the vast majority remain unspoiled. 

\begin{lemma}\label{lem:spoiled-vertices-bound}
    With high probability, there are at most $O(n^{1-2\delta + 4\sigma_L})$ spoiled vertices.
\end{lemma}

\begin{proof}
    First, note that by \Cref{lem:edge-in-shallow-count}, each edge in the queried subgraph of core only appears in $\widetilde{O}(1)$ shallow subgraphs. Hence, $\sum_v |T(v)| \leq O(n^{1-\delta + 2\sigma_L})$ by \Cref{clm:discovered-edges-bound}. Therefore, the total number of vertices whose shallow subgraph contains more than $n^{\delta - 2\sigma_L}$ vertices is $O(n^{1-2\delta+4\sigma_L})$.

    We show that with high probability, there exists at most $O(n^{1-2\delta + 3\sigma_L})$ spoiler vertices in the graph. By \Cref{cor:vertex-in-shallow-count}, since each vertex is in at most $\widetilde{O}(1)$ shallow subgraphs, there are at most $O(n^{1-2\delta + 4\sigma_L})$ spoiled vertices. Thus, it suffices to upper bound the number of spoiler vertices.

    At the time that we add an edge $(u,v)$, the probability that $v$ has a non-zero degree in core is $O(n^{\sigma_L - \delta})$ since by \Cref{clm:discovered-edges-bound}, there are at most $O(n^{1-\delta + \sigma_L})$ vertices with a non-zero degree in the core and by \Cref{clm:incoming-prob}, each of them has a probability of $O(1/n)$ to be the queried edge of $u$. For such an edge, condition (ii) of \Cref{def:spoiler-vertex} holds for vertex $u$ and condition (i) holds for vertex $v$. We assume that during the process of adding edges, for such an edge we count two spoiler vertices (for both endpoints).

    Let $X_i$ be the indicator of having a new spoiler vertex after adding $i$th edge. By the discussion above, we have $\Pr[X_i = 1] \leq O(n^{\sigma_L - \delta})$. Let $k$ be the number of edges found by the algorithm in the core and $X = \sum_{i=1}^k X_i$. Thus, $\E[X] \leq O(n^{1-2\delta + 2\sigma_L})$ since $k = O(n^{1-\delta + \sigma_L})$. Since events are negatively correlated, we get
\begin{align*}
    \Pr\left[|X - \E[X]| \geq 6\sqrt{\E[X] \log n}\right] \leq 2\exp \left(-\frac{(6\sqrt{\E[X] \log n})^2}{3\E[X]} \right) \leq \frac{1}{n^{10}},
\end{align*}
    which implies that there are at most $O(n^{1-2\delta + 3\sigma_L})$ different $i$ such that $X_i = 1$. For each edge, if the indicator is one, we count a constant number of spoiler vertices which concludes the proof.
\end{proof}

\begin{lemma}\label{lem:focs-modification-coupling}
    Let $v$ be a vertex that is not spoiled and belongs to $\{A_r, B_r, D_r\}$. Let $\mathcal{L}(v)$ and $\mathcal{L}'(v)$ be an arbitrary label for $v$ from $\{A_r, B_r, D_r\}$ and the entire queried subgraph of core from all available labels of level $L$ excluding the shallow subgraph of $v$. Then, we have 
    \begin{align*}
        \Pr[T(v) \mid \mathcal{L}(v)] \leq \left(1 + O(n^{\sigma_L-\delta}) \right)^{|T(v)|}\cdot \Pr[T(v) \mid \mathcal{L}'(v)].
    \end{align*}
\end{lemma}

As we have proved, the shallow subgraph of an unspoiled vertex forms a rooted tree. This property allows us to show that all paths starting from the root of this rooted tree and reaching an $S$ vertex or a short cycle, eventually step on a delusive vertex, which, in turn, causes a loss of information about anything below that delusive vertex. Consequently, we can couple the labelings that the tree's root is a vertex of $A_r$ or $B_r$ conditioning on labels of everything outside the shallow subgraph of the root. Hence, the probability that the algorithm queries this exact shallow subgraph no matter what the label of the root is and anything outside of the shallow subgraph. We defer the formal proof of the above lemma to \Cref{sec:focs-modification1} as we extend it to all levels of the construction.

\begin{lemma}\label{lem:prob-of-edge-being-black}
    With high probability, there are at most $O(n^{1-2\delta+5\sigma_L})$ edges $e$ such that $p_e^{inner} > 10n^{\sigma_{L-1} - \sigma_L}$.
\end{lemma}
\begin{proof}
    Let $\widetilde{E}$
    be the set of edges $(u, v)$ (directed from $u$ to $v$) such that $u\in A_r$ that satisfy at least one the following conditions:
    \begin{itemize}
        \item[(i)] $v$ is a spoiled vertex; or
        \item[(ii)] $u$ has at least $n^{\sigma_L}/3$ spoiled neighbors in the queried subgraph of core.
    \end{itemize}
    First, we show $|\widetilde{E}| \leq O(n^{1-2\delta + 5\sigma_L})$. By \Cref{lem:spoiled-vertices-bound}, the number of spoiled vertices is at most $O(n^{1-2\delta + 4\sigma_L})$. Moreover, by \Cref{clm:max-in-degree}, each vertex has at most $\widetilde{O}(1)$ indegree which implies that there are at most $O(n^{1-2\delta + 5\sigma_L})$ edges that satisfy condition (i). On the other hand, if vertex $u$ satisfies the condition (ii), it must have at least $n^{\sigma_L}/4$ edges $(u,w)$ (directed from $u$ to $w$) such that $w$ is spoiled since each vertex has at most $\widetilde{O}(1)$ indegree (\Cref{clm:max-in-degree}). Since the total number of spoiled vertices is $O(n^{1-2\delta + 4\sigma_L})$, there are at most $O(n^{1-2\delta + 5\sigma_L})$ such $u$ that satisfy condition (ii).

    Now, we prove that for all other edges that are not in $\widetilde{E}$, we have that $p_e^{inner} \leq 10n^{\sigma_{L-1} - \sigma_L}$. Since $|\widetilde{E}| \leq O(n^{1-2\delta + 5\sigma_L})$, the aforementioned claim will complete the proof of lemma. For edge $e = (u,v)$ that is directed from $u$ to $v$, if $u \notin A_r$, it is easy to see that $p_e^{inner} = 0$. 
    So assume that $u \in A_r$. Let $v_0 = v, v_1, v_2, \ldots, v_k$ be the neighbors of $u$ in the core in the original graph such that $v_i \in B_r \cup A_r$ and either $v_i$ is a singleton vertex in the queried subgraph or $v_i$ is a directed child of $u$ that is not spoiled. Since $u$ does not satisfy condition (ii), then, $k \geq n^{\sigma_L}/2$. Now we bound the probability that vertex $v_0$ belongs to $A_r$ by using a coupling argument and \Cref{lem:focs-modification-coupling}. 

    Consider a labeling profile $\mathcal{P}$ of all vertices $U = \{v_0, v_1, \ldots, v_k\}$ such that $\mathcal{P}(v_0) = A_r$. By the construction of our input distribution, since $u\in A_r$, at most $O(d_{L-1}) = O(n^{\sigma_{L-1}})$ vertices of $U$ are in $A_r$. We produce $\Omega(n^{\sigma_L})$ new profiles $\mathcal{P}'$ such that $\mathcal{P}'(v_0) \neq A_r$. For each vertex $v_i$ in $U$ such that $\mathcal{P}(v_i) = B_r$, we construct a new profile $\mathcal{P}'$ where $\mathcal{P}(v_j) = \mathcal{P}'(v_j)$ for $j \notin \{0, i\}$, $\mathcal{P}'(v_i) = A_r$, and $\mathcal{P}'(v_0) = B_r$. By \Cref{lem:focs-modification-coupling}, the probability of querying the same shallow subgraphs $T(v_0)$ and $T(v_i)$ in the new labeling profile will be the same up to a factor of 
    $$\left(1 + O(n^{\sigma_L-\delta}) \right)^{|T(v_0)|},$$ 
    and $$\left(1 + O(n^{\sigma_L-\delta}) \right)^{|T(v_i)|},$$
    respectively. Since $v_0$ and $v_i$ are not spoiled vertices, $|T(v_0)|\leq n^{\delta - 2\sigma_L}$ and $|T(v_i)|\leq n^{\delta - 2\sigma_L}$ by \Cref{def:spoiled-vertex}, thus, the probability of having profile $\mathcal{P}$ and $\mathcal{P}'$ are the same up to a factor 
    \begin{align*}
        \left(1 + O(n^{\sigma_L-\delta}) \right)^{|T(v_0)|} \cdot \left(1 + O(n^{\sigma_L-\delta}) \right)^{|T(v_i)|} \leq \left(1 + O(n^{\sigma_L-\delta}) \right)^{2n^{\delta-2\sigma_L}} \leq 1 + o(1).
    \end{align*}
    
    We construct a bipartite graph $H = (P_1, P_2, E_P)$ of labeling profiles such that in $P_1$, we have all profiles $\mathcal{P}$ where $\mathcal{P}(v_0) = A_r$, and in the $P_2$, all profiles $\mathcal{P}'$ where $\mathcal{P}'(v_0) = B_r$. We add an edge between two profiles $\mathcal{P}$ and $\mathcal{P}'$ if we can convert $\mathcal{P}$ to $\mathcal{P}'$ according to the above process. Therefore, $\deg_H(\mathcal{P}) \geq k/2 \geq n^{\sigma_L}/4$ for $\mathcal{P} \in P_1$ since at least $k/2$ vertices of $U$ belong to $B_r$. On the other hand, $\deg_H(\mathcal{P}') \leq 2n^{\sigma_{L-1}}$ for $\mathcal{P}' \in P_2$. To see this, there are at most $2d_{L-1} = 2n^{\sigma_{L-1}}$ vertices $v_i$ in $U$ such that $\mathcal{P}'(v_i) = A_r$ according to the construction of input distribution. Hence,
    \begin{align*}
        p_{(u,v)}^{inner} \leq (1 + o(1))\cdot \frac{|P_1|}{|P_2|} \leq (1 + o(1))\cdot\frac{2n^{\sigma_{L-1}}}{n^{\sigma_L}/4} \leq (1 + o(1))\cdot 8n^{\sigma_{L-1} - \sigma_L} \leq 10n^{\sigma_{L-1} - \sigma_L},
    \end{align*}
    which concludes the proof.
\end{proof}

\subsection{The Algorithm Cannot Create Large Connected Components of Inner Edges}

In this section, we show that as we move downward in the recursive construction, it is harder for the algorithm to create components of large size using edges of the inner level. According to \Cref{lem:prob-of-edge-being-black}, in the highest level of the construction, for at most $O(n^{1-2\delta+5\sigma_L})$ edges in the queried subgraph of the core, the algorithm has the advantage to distinguish that these edges belong to the inner level with probability more than $10n^{\sigma_{L-1} - \sigma_L}$. We assume that the algorithm knows if these edges belong to the inner level or not with probability 1. However, for all other edges that the algorithm queries, it is more likely that those edges belong to the higher level because of the choices of degrees as formalized in \Cref{lem:prob-of-edge-being-black}. More specifically, each other edge that the algorithm queries, has a probability of at most $O(n^{\sigma_{L-1} - \sigma_L})$ to belong to the inner level.  Our goal is to prove a similar lemma to \Cref{lem:prob-of-edge-being-black} for each level in the next two sections. Intuitively, the following lemma shows that as we go down in the recursive construction, the number of edges that the algorithm can distinguish if they belong to the inner level decreases. First, we extend \Cref{def:distinguishing-definition} for all levels in the hierarchy.

\begin{definition}[$p^{\ell-inner}_e$ and Distinguishability of an Edge]\label{def:distinguishing-definition-extension}
    Let $e$ be an edge that is queried by the algorithm. Also, let $p^{\ell-inner}_e$ be the probability that this edge belongs to the subgraph between $A_r^1$ and $A_r^2$ in level $\ell$. We say the algorithm can  distinguish or identify if $e$ belongs to the subgraph between $A_r^1$ and $A_r^2$ if $p^{\ell-inner}_e > 10n^{\sigma_{\ell-1}-\sigma_\ell}$.
\end{definition}

Before proving \Cref{lem:advantage-from-top}, we need to define a function and its characteristics which are crucial to formalize the loss in advantage of the algorithm to identify edges. For the rest of the paper, we define function $g$ for $\ell \in [L]$ as follows 
\begin{align*}
    g(\ell) =(L-\ell+2)\cdot \delta - 5\left(\sum_{i=\ell}^{L} \sigma_{i}/\sigma_{i+1}\right) - 5\left(\sum_{i=\ell}^{L-1}\sigma_i\right)
\end{align*}
where, $\sigma_{L+1}=1$. Also, we let $\sigma_0 = 0$. We have the following observations about the function $g$ that are immediately implied by our choices for $\delta$ and $\sigma_i$ for $i \in [L+1]$.

\begin{observation}\label{obs:g-function-properties}
    The following statements are true regarding function $g$:
    \begin{itemize}
        \item[(i)] $g(\ell - 1) = g(\ell) + \delta - 5\sigma_{\ell - 1}/\sigma_\ell - 5\sigma_{\ell-1}$ for $\ell \in (1, L]$,
        \item[(ii)] $1-g(\ell - 1) - 3\sigma_{\ell-1} > 1 - g(\ell) - \delta + 5\sigma_{\ell-1}/\sigma_\ell + \sigma_{\ell-1}$ for $\ell \in (1, L]$,
        \item[(iii)] $g(1) > 2$,
        \item[(iv)] $1 - g(\ell) \neq 0$ for all $\ell \in [L]$.
    \end{itemize}
\end{observation}
\begin{proof}
    \begin{itemize}
        \item[(i):] By the definition of function $g$, we have
        \begin{align*}
            g(\ell - 1) &= (L-\ell+3)\cdot \delta - 5\left(\sum_{i=\ell-1}^{L} \sigma_{i}/\sigma_{i+1}\right) - 5\left(\sum_{i=\ell-1}^{L-1}\sigma_i\right) \\
         & = \left[(L-\ell + 2)\cdot \delta -  5\left(\sum_{i=\ell}^{L} \sigma_{i}/\sigma_{i+1}\right) - 5\left(\sum_{i=\ell}^{L-1}\sigma_i\right) \right] + \delta - 5\sigma_{\ell-1}/\sigma_\ell - 5\sigma_{\ell-1}\\
         & = g(\ell) + \delta - 5\sigma_{\ell-1}/\sigma_\ell - 5\sigma_{\ell-1}.
        \end{align*}
        \item[(ii):] By statement (i), we get
        \begin{align*}
            1-g(\ell - 1) - 3\sigma_{\ell-1} =  1 - g(\ell) - \delta + 5\sigma_{\ell-1}/\sigma_\ell + 2\sigma_{\ell-1}  > 1 - g(\ell) + 5\sigma_{\ell-1}/\sigma_\ell + \sigma_{\ell-1}.
        \end{align*}
        \item [(iii):] By the definition of function $g$, we have
        \begin{align*}
            g(1) &= (L+1)\cdot \delta - 5\left(\sum_{i=1}^{L} \sigma_{i}/\sigma_{i+1}\right) - 5\left(\sum_{i=1}^{L-1}\sigma_i\right)\\
            & = (L+1)\cdot \delta - \left(\frac{\delta L}{2} \right) - 5\left(\sum_{i=1}^{L-1} (\frac{\delta}{10})^{L+1-i}\right) & (\text{By \Cref{tbl:parameters}})\\
            & > (L+1)\cdot \delta - \left(\frac{\delta L}{2} \right) - \delta\\
            & = 2  & (\text{Since } L = 4/\delta).
        \end{align*}
        \item[(iv):] If $g(\ell)$ is zero for a particular $\ell$, we can perturb the parameters in \Cref{tbl:parameters} to meet all constraints and make g(l) non-zero.
    \end{itemize}
\end{proof}

\begin{lemma}\label{lem:advantage-from-top}
    With high probability, the following statements hold:
    \begin{itemize}
         \item[(i)] If $1-g(\ell) < 0$, then with probability $1 - O(n^{1-g(\ell)})$, there exist no edge $e$ such that $p_e^{\ell-inner} > 10n^{\sigma_{\ell-1} - \sigma_\ell}$. Also, with high probability, there are at most $\widetilde{O}(1)$ edges $e$ such that $p_e^{\ell-inner} > 10n^{\sigma_{\ell-1} - \sigma_\ell}$.
         \item[(ii)] If $1-g(\ell) > 0$, with high probability, there are at most $O(n^{1-g(\ell)})$ edges $e$ such that $p_e^{\ell-inner} > 10n^{\sigma_{\ell-1} - \sigma_\ell}$.
     \end{itemize}
    
\end{lemma}

Note that if we replace $\ell = L$ in the above bound, we get the same bound as \Cref{lem:prob-of-edge-being-black}. We use $E^{inner}_\ell$ to show the set of edges that $p_e^{\ell-inner} > 10n^{\sigma_{\ell-1} - \sigma_\ell}$. If the algorithm can distinguish the difference between a graph from \yesdist{} and a graph from \nodist{}, it should be able to distinguish between the subgraphs between $A_r^1$ and $A_r^2$ of level $\ell$ as other parts of the two graphs are similar. In this paper, when we mention the inner level, we only mean the subgraph between $A_r^1$ and $A_r^2$ of that level. In this section, we denote the edges between $A_r^1$ and $A_r^2$ of level $\ell$ as {\em black edges} and we denote other edges as {\em green edges}. We prove that the algorithm cannot grow a large component of black edges. The following lemma is the main technical contribution of this section.

\begin{lemma}\label{lem:small-square-component}
     Let $C_1, C_2, \ldots, C_c$ be the underlying undirected connected components of black edges where there exists at least one edge of $E^{inner}_\ell$ in each of the components. Then, the following statements hold:
     \begin{itemize}
         \item[(i)] If $1-g(\ell) < 0$, then $c=0$ with probability $1-O(n^{1-g(\ell)})$. Also, $\sum_{i=1}^c |C_i| \leq \widetilde{O}( n^{5\sigma_{\ell-1}/\sigma_{\ell}})$ with high probability.
         \item[(ii)] If $1-g(\ell) > 0$, we have $\sum_{i=1}^c |C_i| \leq O(n^{1-g(\ell) + 5\sigma_{\ell-1}/\sigma_\ell})$ with high probability.
     \end{itemize}
\end{lemma}

We use induction to show the correctness of \Cref{lem:advantage-from-top} and \Cref{lem:small-square-component}. For the base case, we already proved that \Cref{lem:advantage-from-top} holds when $\ell = L$ (\Cref{lem:prob-of-edge-being-black}). To prove \Cref{lem:small-square-component} for a fix $\ell$, we use the bound from \Cref{lem:advantage-from-top} for $\ell$. Then, we use the result to prove \Cref{lem:advantage-from-top} for $\ell - 1$. In this section, we focus on proving \Cref{lem:small-square-component} using \Cref{lem:advantage-from-top}. In the rest of this subsection, we focus on the step to prove \Cref{lem:small-square-component}.

With the same argument as \Cref{clm:discovered-edges-bound}, we can give an upper bound for the number of black edges which is formalized in \Cref{clm:black-edges-count}.

\begin{claim}\label{clm:black-edges-count}
There are at most $O(n^{1-\delta + \sigma_{\ell-1}})$ black edges with high probability.
\end{claim}
\begin{proof}
    The proof is similar to the proof of \Cref{clm:discovered-edges-bound}. We repeat the argument for completeness. For all vertices to which the algorithm makes more than $\tau n_L/2$ adjacency list queries, we assume that it discovers all its black edges. Since the algorithm makes at most $O(n^{2-\delta})$ queries in total, the total number of vertices with more than $\tau n_L /2$ queries cannot be larger than $O(n^{1-\delta})$ and therefore, the total discovered black edges incident to these vertices is at most $O(d_{\ell-1} \cdot n^{1-\delta}) = O(n^{1-\delta + \sigma_{\ell - 1}})$.

    For all other vertices, each adjacency list query is a black edge with a probability of $O(d_{\ell-1} / n) = O(n^{\sigma_{\ell-1}}/n)$. Since there are $O(n^{2-\delta})$ queries in total, with high probability, the algorithm will find at most $O(n^{1-\delta + \sigma_{\ell - 1}})$ black edges using a Chernoff bound.
\end{proof}

According to \Cref{lem:advantage-from-top}, for all edges excluding those in $E^{inner}_\ell$, when the algorithm queries an edge, it has a higher probability of being a green edge. This intuitively implies that, for any given vertex $u$, the algorithm should not be capable of discovering numerous descendants that are exclusively reachable through directed black edges, provided we disregard edges in $E^{inner}_\ell$.

\begin{lemma}\label{lem:desendants-upper-bound}
Consider all queried black edges in the core except edges $E^{inner}_\ell$. With high probability, each vertex has at most $n^{5\sigma_{\ell-1} / \sigma_\ell}$ descendants that are reachable by directed black edges. Moreover, for each vertex, the total number of black edges to all its descendants is at most $n^{5\sigma_{\ell-1} / \sigma_\ell}$.
\end{lemma}
\begin{proof}
    Fix a vertex $u$. First, we claim that the probability of having a directed path of length $i$ that starts from $u$ and ends in a vertex $v$ is bounded by $n^{i(\sigma_{\ell-1} - \sigma_{\ell})/2}$. We use induction to prove this claim. For the base case where $i = 1$, if there is no edge between $u$ and $v$ this probability is 0. If there exists an edge, by \Cref{lem:advantage-from-top}, this edge is black with probability of at most $10n^{\sigma_{l-1} - \sigma_l} < n^{(\sigma_{l-1} - \sigma_l)/2}$. Suppose that the claim holds for all $i' < i$. By \Cref{clm:max-in-degree}, vertex $v$ has at most $5\log n$ indegree in the whole queried subgraph (including all edges). Let $\{v_1, v_2, \ldots, v_k\}$ be the set of vertices that have directed edge to $v$. Thus, if there exists a directed black path of length $i$ to $v$, there must exist a path of length $i-1$ to one of $v_j$ and a black edge from $v_j$ to $v$. Let $B^i_w$ be the event that there exists a directed black path of length $i$ to vertex $w$. Using a union bound,
    \begin{align*}
        \Pr[B^i_v] &\leq \sum_{j=1}^k \Pr[B^{i-1}_{v_j}] \cdot \Pr[(v_j, v) \text{ is black}]\\
        & \leq \sum_{j=1}^k \Pr[B^{i-1}_{v_j}] \cdot 10n^{\sigma_{l-1} - \sigma_l} & (\text{By  \Cref{lem:advantage-from-top}})\\
        & \leq k \cdot n^{(i-1)(\sigma_{l-1} - \sigma_l)/2} \cdot 10n^{\sigma_{l-1} - \sigma_l} & (\text{Induction hypothesis})\\
        & \leq 50 \cdot \log n \cdot n^{(i+1)(\sigma_{l-1} - \sigma_l)/2} & (k \leq 5 \log n \text{ by  \Cref{clm:max-in-degree}})\\
        & \leq n^{i(\sigma_{l-1} - \sigma_l)/2},
    \end{align*}
    which completes the induction step.

    Second, we show that there is no directed black path of length $5/\sigma_l$ with high probability in the graph. To see this, the probability of having a directed black path of length $5/\sigma_l$ between two vertices $u$ and $v$ is upper bounded by $n^{5(\sigma_{l-1} - \sigma_l)/(2\sigma_l)}$. Taking a union bound over all possible pairs, we obtain
    \begin{align*}
        \Pr\left[\exists \text{ directed black path of length $5/\sigma_l$}\right] & \leq n^2 \cdot n^{5(\sigma_{l-1} - \sigma_l)/(2\sigma_l)}\\
        & \leq n^2 \cdot n^{-9/4} & (\text{Since } \sigma_{l-1} < \frac{\sigma_l}{10})\\
        & \leq n^{-1/4}.
    \end{align*}
    Therefore, we can assume that with high probability there is no directed black path of length $5/\sigma_l$ in the queried subgraph.

    Finally, suppose that we condition on not having a directed black path of length $5/\sigma_l$. Since each vertex has at most $n^{\sigma_{l-1}}$ black edges in total (even not queried by the algorithm), the total number of vertices and edges that are reachable up to distance $5/\sigma_l$ from a fixed vertex $u$ is upper bounded by $n^{5\sigma_{\ell-1} / \sigma_\ell}$.
\end{proof}

\begin{corollary}\label{cor:longest-black-directed-path}
    Consider all queried black edges in the core except edges $E^{inner}_\ell$. The longest directed path of black edges has length at most $5/\sigma_\ell$.
\end{corollary}
\begin{proof}
    The proof follows by the proof of \Cref{lem:desendants-upper-bound}.
\end{proof}

It is important to observe that the algorithm discovers black incident edges for only a small fraction of vertices when compared to the total number of vertices.  Additionally, if we exclude $E^{inner}_\ell$, the size of the black descendants of each vertex is constrained as indicated in \Cref{lem:desendants-upper-bound}. Consequently, we anticipate a limited number of intersections between the descendants of vertices. This insight is further formalized in the following claims and corollary.

Let $SCC_1, SCC_2, \ldots, SCC_s$ be the strongly connected components of directed black edges that are queried by the algorithm. For each component such that its indegree is zero (roots of the directed acyclic graph of strongly connected components), we choose a vertex to represent the component. Let $R = \{u_1, u_2, \ldots, u_{s
'}\}$ be the set of the chosen vertices. Note that each vertex $v \notin R$, is in a black descendent of at least one of the vertices in $R$.

\begin{claim}\label{clm:intersect-desendants-prob}
    Consider all queried black edges in the core except edges $E^{inner}_\ell$. Let $v \in R$. Then, the probability that there exists a vertex $u \in R \setminus \{ v\}$ such that $u$'descendants intersect $v$'s descendants is at most $O(n^{5\sigma_{\ell-1} / \sigma_\ell - \delta + \sigma_{\ell-1}})$.
\end{claim}

\begin{proof}
    By \Cref{lem:desendants-upper-bound}, vertex $v$ has at most $n^{5\sigma_{l-1}/\sigma_l}$ descendants. Combining with \Cref{clm:incoming-prob}, the probability that each new query goes to a vertex that is descendant of $v$ is at most $O(n^{5\sigma_{l-1}/\sigma_l}/n)$. Since the total number of black edges is upper bounded by $O(n^{1-\delta + \sigma_{l-1}})$, then the probability that there exists a vertex $u \in R \setminus \{ v\}$ such that $u$'descendants intersect descendants $v$'s descendants is at most $O(n^{5\sigma_{\ell-1} / \sigma_\ell - \delta + \sigma_{\ell-1}})$ using a union bound.
\end{proof}

\begin{corollary}\label{cor:intersect-desendants-prob}
    Consider all queried black edges in the core except edges $E^{inner}_\ell$. Let $k < 50\sigma_{l}/\sigma_{l-1}$ and $v_1, \ldots, v_k$ be $k$ arbitrary vertices in $R$. Then, the probability that there exists a vertex $u \in R \setminus \{v_1, \ldots, v_k\}$ such that $u$'descendants intersect descendants of vertices in $\{v_1, \ldots, v_k\}$ is at most $O(n^{5\sigma_{\ell-1} / \sigma_\ell - \delta + \sigma_{\ell-1}})$.
\end{corollary}
\begin{proof}
    The proof follows the same as proof of \Cref{clm:intersect-desendants-prob} and the fact that $k$ is a constant.
\end{proof}

Hence, we anticipate these black connected components to have a very small size, given that the number of descendants for each vertex is quite limited, and the chances of their intersection are low.

\begin{claim}\label{clm:small-component-in-intersect-tree}
     Consider all queried black edges except edges $E^{inner}_\ell$. Let $C$ be an arbitrary connected component of these edges. Then, with high probability $|C| \leq O(n^{5\sigma_{\ell-1} / \sigma_\ell})$. Moreover, each connected component has at most $O(|C|)$ black edges.
\end{claim}

\begin{proof}
    We construct a new graph $H$ with the same vertex set $R$. We add edge $(u,v)$ to $H$ if black descendants of $u$ and $v$ intersect. In what follows, we prove that the largest connected component of $H$ is at most $\varrho = 10/\delta$ with high probability. Since the number of black descendants (and black edges to its descendants) for each vertex is at most $n^{5\sigma_{l-1}/\sigma_l}$ by \Cref{lem:desendants-upper-bound}, this claim is enough to finish the proof. 

    Suppose that we start the following process from vertex $v \in R$. In the beginning, we have a set $C$ that only contains $v$.  In each step, we reveal one of the edges from vertices in $C$ to vertices in $R \setminus C$. Assume that this edge is $(w,z)$ where $w \in C$ and $z \in R \setminus C$. We add $z$ to $C$ and continue the process. The process stops either when there is no edge from $C$ to $R \setminus C$ or when $|C| > \varrho$. Let $X_i$ be the event that there exists an edge between $C$ and $R \setminus C$ in step $i$ of the process. By \Cref{cor:intersect-desendants-prob}, we have $\Pr[X_i = 1 \mid X_1, \ldots, X_{i-1}] \leq O(n^{5\sigma_{\ell-1} / \sigma_\ell - \delta + \sigma_{\ell-1}})$. Let $Y_v$ be the event that the process stops when $|C| > \varrho$. Hence, 
    \begin{align*}
        \Pr[Y_v] & = \prod_{i=1}^{i \leq \varrho} \Pr[X_i \mid X_1, X_2, \ldots, X_{i-1}]\\
        & \leq O\left( (n^{5\sigma_{\ell-1} / \sigma_\ell - \delta + \sigma_{\ell-1}})^{\varrho} \right)\\
        & = O\left(\frac{1}{n^{5}}\right) & (\text{Since } \varrho = 10/\delta). 
    \end{align*}
    Therefore, using a union bound over all possible $v \in R$, with a probability of $1 - O(n^{-4})$, there is no connected component of size larger than $\varrho$ in $H$. 
\end{proof}

\begin{corollary}\label{cor:small-component-in-intersect-tree-number}
     Consider all queried black edges except edges $E^{inner}_\ell$. Let $C$ be an arbitrarily connected component of these edges that is created by the intersection of descendants of $\varrho$ vertices in $R$. Then, with high probability $\varrho \leq 10/\delta$.
\end{corollary}
\begin{proof}
    The proof follows by the proof of \Cref{clm:small-component-in-intersect-tree}.
\end{proof}

We now possess all the necessary tools to establish \Cref{lem:small-square-component}. At a high level, we assume that the algorithm has control over where to put the edges of $E^{inner}_\ell$ to maximize $\sum_{i=1}^c |C_i|$. However, we show that even by giving this power to the algorithm, we can still prove the bound stated in the lemma.

\begin{proof}[Proof of \Cref{lem:small-square-component}]
    Suppose that an adversary chooses how edges of $E^{inner}_\ell$ are between the components. Let $\widehat{\mathcal{C}} = \{\widehat{C_1}, \widehat{C_2}, \ldots, \widehat{C_{c'}}\}$ be the connected components before adding edges $E^{inner}_\ell$ by the adversary. First, note that $|\widehat{C_i}| < O(n^{5\sigma_{\ell-1} / \sigma_\ell})$ for all $i \in [c']$ with high probability by \Cref{clm:small-component-in-intersect-tree}.

    Let $C_1, C_2, \ldots, C_{c}$ be the connected components after adding edges $E^{inner}_\ell$ and removing the components that do not have any of the edges in $E^{inner}_\ell$. Each edge of $E^{inner}_\ell$ can connect at most two components of $\widehat{\mathcal{C}}$. Therefore, the total number of components in $\widehat{\mathcal{C}}$ that have at least one edge of $E^{inner}_\ell$ is upper bounded by $O(|E^{inner}_\ell|)$. Now if $1-g(\ell) < 0$, according to the statement (i) of \Cref{lem:advantage-from-top}, with probability $1-O(n^{1-g(\ell)})$ we have $|E^{inner}_\ell| = 0$. Also, with a high probability $|E^{inner}_\ell| = \widetilde{O}(1)$. Combining with $|\widehat{C_i}| < O(n^{5\sigma_{\ell-1} / \sigma_\ell})$, we obtain the proof of statement (i).

    If $1-g(\ell) > 0$, according to the statement (ii) of \Cref{lem:advantage-from-top}, we have $|E^{inner}_\ell| \leq O(n^{1-g(\ell)})$ with high probability. Combining with $|\widehat{C_i}| < O(n^{5\sigma_{\ell-1} / \sigma_\ell})$, we obtain $\sum_{i=1}^c |C_i| \leq O(n^{1-g(\ell) + 5\sigma_{\ell-1}/\sigma_\ell})$ which concludes the proof of (ii).
\end{proof}

\subsection{Smaller Connected Components Results in Less Identified Inner Edges}

In this section, we use \Cref{lem:small-square-component} to show that as the size of connected components gets smaller, it is harder for the algorithm to identify black edges. We abuse the notation to generalize the definition of spoiled vertex and shallow subgraph similar to the warm-up section.

\begin{definition}[$\ell$-Shallow Subgraph]\label{def:shallow-subgraph-extend}
Suppose that we define green and black edges with respect to level $\ell$ and $\ell - 1$ of the construction hierarchy. For a vertex $v$, we let the $\ell$-shallow subgraph of $v$ be a set of vertices that are reachable by $v$ within a distance of $10\log n$ using directed paths with only black edges from $v$ in the queried subgraph. We use $T^\ell(v)$ to denote the $\ell$-shallow subgraph of $v$.
\end{definition}

With the exact same proof as \Cref{lem:edge-in-shallow-count}, we can extend its claim to $\ell$-Shallow Subgraph.

\begin{lemma}\label{cor:vertex-in-shallow-count-extended}
    With high probability, each vertex is in at most $\widetilde{O}(1)$ $\ell$-shallow subgraphs.
\end{lemma}

\begin{corollary}\label{lem:edge-in-shallow-count-extended}
    With high probability, each black edge that the algorithm finds is in at most $\widetilde{O}(1)$ $\ell$-shallow subgraphs.
\end{corollary}

\begin{observation}\label{obs:number-edges-black-comp}
    Let $C_1, C_2, \ldots, C_c$ be the underlying undirected connected components of black edges, and let $E(C_i)$ be the edges set of component $C_i$. Then, $|E(C_i)| \leq O(\log n)\cdot |C_i| $.
\end{observation}
\begin{proof}
    The proof follows by the fact that each vertex has an incoming degree of at most $5\log n$ in the whole queried subgraph of core by \Cref{clm:max-in-degree}. 
\end{proof}

\begin{observation}\label{obs:total-number-edges-black-comp}
    Let $C_1, C_2, \ldots, C_c$ be the underlying undirected connected components of black edges where there exists at least one edge of $E^{inner}_\ell$ in each of the components. Let $E(C_i)$ denote the edge set of component $C_i$. Then, the following statements hold:
    \begin{itemize}
        \item[(i)] If $1-g(\ell) < 0$, then $c=0$ with probability $1-O(n^{1-g(\ell)})$. Also, with a high probability, $\sum_{i=1}^c |E(C_i)| \leq \widetilde{O}(n^{5\sigma_{\ell-1}/\sigma_\ell})$
        \item[(ii)] If $1-g(\ell) > 0$, we have $\sum_{i=1}^c |E(C_i)| \leq \widetilde{O}(n^{1-g(\ell) + 5\sigma_{\ell-1}/\sigma_\ell})$ with high probability.
    \end{itemize}
\end{observation}

\begin{proof}
    Combining each statement of \Cref{lem:small-square-component} and \Cref{obs:number-edges-black-comp} yields each statement.
\end{proof}

\begin{definition}[$\ell$-Spoiler Vertex]\label{def:spoiler-vertex-extend}
    For $\ell \in (1, L]$, let $\hat{E}$ be the set of black edges that are in a connected component with at least one edge of $E^{inner}_\ell$. Let $u$ be a vertex that is in a black connected component that contains at least one edge of $E^{inner}_\ell$. We say a vertex $u$ in the core is {\em $\ell$-spoiler} if at least one of the following conditions holds:
    \begin{itemize}
        \item[(i)] vertex $u$ has more than one incoming edge, 
        \item[(ii)] there is an edge $(u,v) \in \hat{E}$ that is discovered by the algorithm at a time when $v$ already has non-zero degree.
    \end{itemize}
\end{definition}

\begin{definition}[$\ell$-Spoiled Vertex]\label{def:spoiled-vertex-extend}
    For $\ell \in (1, L]$, let $v$ be a vertex that is in a black connected component that contains at least one edge of $E^{inner}_\ell$. Then, vertex $v$ is $\ell$-spoiled if its $\ell$-shallow subgraph contains any of the following:
    \begin{itemize}
        \item a $\ell$-spoiler vertex; or
        \item at least $n^{\delta - 2\sigma_L}$ vertices.
    \end{itemize}
\end{definition}

\begin{observation}\label{obs:tree-structure-of-unspoiled-extension}
    Let $v$ be a vertex that is not $\ell$-spoiled. Then, the $\ell$-shallow subgraph of $v$ is a rooted tree of size at most $n^{\delta - 2\sigma_L}$. Moreover, for each edge $(u,w)$ in the $\ell$-shallow subgraph of $v$, at the time that the algorithm made the query, $w$ was a singleton vertex.
\end{observation}
\begin{proof}
    Because of \Cref{def:spoiler-vertex-extend} and \Cref{def:spoiled-vertex-extend}, when the algorithm finds an edge $(u, w)$ in the $\ell$-shallow subgraph of $v$, the other endpoint must be singleton which implies that the $\ell$-shallow subgraph of $v$ is a rooted tree.
\end{proof}

In the next two claims, we provide a bound on probability and the number of vertices that do not satisfy the second condition in \Cref{def:spoiled-vertex-extend}.

\begin{claim}\label{clm:bounding-ell-spoiled-count}
    Suppose that $1-g(\ell-1) - 3\sigma_{\ell-1} \geq 0$. With high probability, there are at most $O(n^{1-g(\ell-1)-2\sigma_{\ell-1}})$ vertices $v$ where:
    \begin{itemize}
        \item there exist an edge of $E^{inner}_\ell$ in their black connected component; and
        \item $|T^\ell(v)| > n^{\delta-2\sigma_{\ell-1}}$.
    \end{itemize}
\end{claim}
\begin{proof}
    Let $C_1, C_2, \ldots, C_c$ be the underlying undirected connected components of black edges where there exists at least one edge of $E^{inner}_\ell$ in each of the components. Also, let $\widehat{V}$ be the set of vertices in these components. Let $E(C_i)$ be the set of edges of component $i$. By statement (ii) of \Cref{obs:total-number-edges-black-comp}, we have $\sum_{i=1}^c |E(C_i)| \leq \widetilde{O}(n^{1-g(\ell) + 5\sigma_{\ell-1}/\sigma_\ell})$. Applying \Cref{lem:edge-in-shallow-count-extended}, we obtain
    \begin{align*}
        \sum_{u \in \widehat{V}} |T^\ell(u)| \leq \widetilde{O}(1) \cdot \sum_{i=1}^c |E(C_i)| \leq \widetilde{O}(n^{1-g(\ell)+5\sigma_{\ell-1}/\sigma_{\ell}}).
    \end{align*}
    Let $\varrho$ denote the number of vertices $v$ where $|T^\ell(v)| > n^{-g(\ell)-10\sigma_{\ell-1}/\sigma_\ell}$. Therefore,
    \begin{align*}
        \varrho \leq \frac{\sum_{u \in \widehat{V}} |T^\ell(u)|}{n^{\delta-2\sigma_{\ell-1}}} \leq \frac{\widetilde{O}(n^{1-g(\ell)+5\sigma_{\ell-1}/\sigma_{\ell}})}{n^{\delta-2\sigma_{\ell-1}}} \leq \widetilde{O}(n^{1-g(\ell)-\delta+5\sigma_{\ell-1}/\sigma_{\ell}+2\sigma_{\ell-1}}) \leq O(n^{1-g(\ell-1)-2\sigma_{\ell - 1}})
    \end{align*}
    where the last inequality is followed by statement (ii) of \Cref{obs:g-function-properties}.
\end{proof}

\begin{claim}\label{clm:bounding-ell-spoiled-prob}
    Suppose that $1-g(\ell-1)-3\sigma_{\ell-1}< 0$. With high probability, there exists no vertex such that
    \begin{itemize}
        \item there exist an edge of $E^{inner}_\ell$ in their black connected component; and
        \item $|T^\ell(v)| > n^{\delta-2\sigma_{\ell-1}}$.
    \end{itemize}
\end{claim}
\begin{proof}
    Let $C_1, C_2, \ldots, C_c$ be the underlying undirected connected components of black edges where there exists at least one edge of $E^{inner}_\ell$ in each of the components. Also, let $\widehat{V}$ be the set of vertices in these components. First, if $1-g(\ell) < 0$, with high probability we have $\sum_{i=1}^c |E(C_i)| \leq \widetilde{O}(n^{5\sigma_{\ell-1}/\sigma_{\ell}})$ by statement (i) of \Cref{obs:total-number-edges-black-comp}. Since $\delta - 2\sigma_{\ell-1} > 5\sigma_{\ell-1}/\sigma_\ell$, there is no component with an edge from $E^{inner}_\ell$ with size $n^{\delta - 2\sigma_{\ell-1}}$ with high probability.

    Next, if $1-g(\ell) > 0$, with high probability, we have $\sum_{i=1}^c |E(C_i)| \leq \widetilde{O}(n^{1-g(\ell) + 5\sigma_{\ell-1}/\sigma_{\ell}})$ by statement (ii) of \Cref{obs:total-number-edges-black-comp}. Applying \Cref{lem:edge-in-shallow-count-extended}, we obtain
    \begin{align*}
        \sum_{u \in \widehat{V}} |T^\ell(u)| \leq \widetilde{O}(1) \cdot \sum_{i=1}^c |E(C_i)| \leq \widetilde{O}(n^{1-g(\ell)+5\sigma_{\ell-1}/\sigma_{\ell}}).
    \end{align*}
    Moreover, 
    \begin{align*}
        1-g(\ell) + 5\sigma_{\ell-1}/\sigma_{\ell} &= 1 - g(\ell-1) + \delta - 5\sigma_{\ell-1} & (\text{By statement (i) of \Cref{obs:g-function-properties}}) \\
        & = \left(1 - g(\ell - 1) - 3\sigma_{\ell-1} \right) + \left(\delta - 2\sigma_{\ell-1} \right)\\
        & < \delta - \sigma_{\ell-1},
    \end{align*}
    where the last inequality follows by the assumption that $1-g(\ell-1)-3\sigma_{\ell-1}/\sigma_\ell < 0$. Therefore, with a high probability, there is no component with an edge from $E^{inner}_\ell$ with size $n^{\delta - \sigma_{\ell-1}}$ with high probability.
\end{proof}

Just as in \Cref{lem:spoiled-vertices-bound}, we can give bounds on the probability and the count of $\ell$-spoiled vertices. While the proof steps closely resemble those in the warm-up section, for the sake of thoroughness, we reiterate some of the key arguments.

\begin{lemma}\label{lem:ell-spoiled-count}
    Suppose that $1-g(\ell-1)-3\sigma_{\ell-1} \geq 0$. With high probability, there are at most $O(n^{1-g(\ell-1)-2\sigma_{\ell-1}})$ $\ell$-spoiled vertices.
\end{lemma}

\begin{proof}
    The proof has the same steps as \Cref{lem:spoiled-vertices-bound}. First, by \Cref{clm:bounding-ell-spoiled-count}, with high probability there are at most $O(n^{1-g(\ell-1)-2\sigma_{\ell-1}})$ vertices $v$ in a component with at least one edge of $E^{inner}_\ell$ such that $|T^\ell(v)| > n^{\delta - \sigma_{\ell-1}}$. Let $C_1, C_2, \ldots, C_c$ be the underlying undirected connected components of black edges where there exists at least one edge of $E^{inner}_\ell$ in each of the components. Let $\widehat{E}$ be the set of black edges of these components. By statement (ii) of \Cref{lem:small-square-component}, $|\widehat{E}| \leq O(n^{1-g(\ell) + 5\sigma_{\ell-1}/\sigma_\ell})$.

    Next, suppose that we add edges $\widehat{E}$ that are queried by the algorithm in the same order as the algorithm queried them. We show that with high probability, there exists at most $O(n^{1-g(\ell-1)-4\sigma_{\ell-1}})$ $\ell$-spoiler vertices in the graph. By \Cref{cor:vertex-in-shallow-count-extended}, since each vertex is in at most $\widetilde{O}(1)$ $\ell$-shallow subgraphs, then there are at most $O(n^{1-g(\ell-1)-3\sigma_{\ell-1}})$  $\ell$-spoiled vertices. So in the rest, we focus on upper bounding the number of $\ell$-spoiler vertices.

    At the time that we add an edge $(u,v)$, the probability that $v$ has at least one black edge is $O(n^{\sigma_{\ell-1} - \delta})$ since by \Cref{clm:black-edges-count}, there are at most $O(n^{1-\delta + \sigma_{\ell-1}})$ vertices with a black edge and by \Cref{clm:incoming-prob}, each of them has a probability of $O(1/n)$ to be the queried edge of $u$. For such an edge, condition (ii) holds for vertex $u$ and condition (i) holds for vertex $v$. We assume that during the process of adding edges, for such an edge we count two spoiler vertices (for both endpoints).

    Let $X_i$ be the indicator of having a new spoiler vertex after adding $i$th edge. By the discussion above, we have $\Pr[X_i = 1] \leq O(n^{\sigma_{\ell-1} - \delta})$. Let $X = \sum_{i=1}^{|\widehat{E}|} X_i$. Thus, \begin{align*}
        \E[X] \leq O(n^{1-g(\ell) + 5\sigma_{\ell-1}/\sigma_\ell - \delta + \sigma_{\ell-1}}), 
    \end{align*}
    since $|\widehat{E}| \leq O(n^{1-g(\ell) + 5\sigma_{\ell-1}/\sigma_\ell})$. Since events are negatively correlated, we get
\begin{align*}
    \Pr\left[|X - \E[X]| \geq 6\sqrt{\E[X] \log n}\right] \leq 2\exp \left(-\frac{(6\sqrt{\E[X] \log n})^2}{3\E[X]} \right) \leq \frac{1}{n^{10}},
\end{align*}
    which implies that there are at most $O(n^{1-g(\ell) + 5\sigma_{\ell-1}/\sigma_\ell - \delta + \sigma_{\ell-1}})$ different $i$ such that $X_i = 1$. For each edge, if the indicator is one, we count a constant number of spoiler vertices. Moreover, by statement (i) of \Cref{obs:g-function-properties},
    \begin{align*}
        1-g(\ell) + 5\sigma_{\ell-1}/\sigma_\ell - \delta + \sigma_{\ell-1} & = 1 - g(\ell - 1) - 4\sigma_{\ell-1},
    \end{align*}
    which concludes the proof.
\end{proof}

\begin{lemma}\label{lem:ell-spoiled-prob}
    Suppose that $1-g(\ell-1)-3\sigma_{\ell-1} < 0$. Then, with probability of $1 - O(n^{1-g(\ell-1)-3\sigma_{\ell - 1}})$, there is no $\ell$-spoiled vertex. Moreover, with a high probability, there are at most $\widetilde{O}(1)$ $\ell$-spoiled vertices.
\end{lemma}

\begin{proof}
    Because of the assumption that $1-g(\ell-1)-3\sigma_{\ell-1} < 0$,  with high probability there is no vertex $v$ in a component with at least one edge of $E^{inner}_\ell$ such that $|T^\ell(v)| > n^{\delta - 2\sigma_{\ell-1}}$ by \Cref{clm:bounding-ell-spoiled-prob}. Let $C_1, C_2, \ldots, C_c$ be the underlying undirected connected components of black edges where there exists at least one edge of $E^{inner}_\ell$ in each of the components. Let $\widehat{E}$ be the set of black edges of these components. We consider two possible scenarios:

    \paragraph{(Case 1) $1-g(\ell) > 0$:} in this case, we have $\sum_{i=1}^c |C_i| \leq \widetilde{O}(n^{1-g(\ell) + 5\sigma_{\ell-1}/\sigma_\ell})$ with high probability according to statement (ii) of \Cref{obs:total-number-edges-black-comp}. We prove that with probability $1-O(n^{1-g(\ell-1)-3\sigma_{\ell - 1}})$, there exists no $\ell$-spoiler vertex. Suppose that we add edges of $\widehat{E}$ according to the ordering that the algorithm queried them. With the exact same argument as proof of \Cref{lem:ell-spoiled-count}, each edge that we add has a probability of $O(n^{\sigma_{\ell-1}-\delta})$ to create a constant number of $\ell$-spoiler vertices. Using a union bound, the probability of having a $\ell$-spoiler vertex is bounded by 
    \begin{align*}
        |\widehat{E}| \cdot O(n^{\sigma_{\ell-1}-\delta}) \leq \widetilde{O}(n^{1-g(\ell) + 5\sigma_{\ell-1}/\sigma_\ell+\sigma_{\ell-1}-\delta}) \leq \widetilde{O}(n^{1-g(\ell-1)-4\sigma_{\ell - 1}}),
    \end{align*}
    where the last inequality is followed by statement (i) of \Cref{obs:g-function-properties}. On the other hand, since the expected number of $\ell$-spoiler vertices is less than 1, using a Chernoff bound we can show that with high probability there are at most $\widetilde{O}(1)$ $\ell$-spoiler vertices.

    \paragraph{(Case 2) $1-g(\ell) < 0$:} in this case, according to the statement (i) of \Cref{obs:total-number-edges-black-comp}, there is no component with an edge of $E^{inner}_\ell$ with probability $1-O(n^{1-g(\ell)})$ and therefore, there is no $\ell$-spoiled vertex with probability of $O(n^{1-g(\ell)})$. Now suppose that we condition on having a component with an edge of $E^{inner}_\ell$. By statement (i) of \Cref{obs:total-number-edges-black-comp}, we have $|\widehat{E}| \leq \widetilde{O}(n^{5\sigma_{\ell-1}/\sigma_\ell})$ with high probability. Similar to the previous case, the probability of having a $\ell$-spoiler vertex is bounded by 
    \begin{align*}
        |\widehat{E}| \cdot O(n^{\sigma_{\ell-1}-\delta}) \leq \widetilde{O}(n^{5\sigma_{\ell-1}/\sigma_\ell+\sigma_{\ell-1}-\delta}).
    \end{align*}
    Since the probability of having a component with an edge of $E^{inner}_\ell$ is $O(n^{1-g(\ell)})$, the probability of having a $\ell$-spoiler vertex is upper bounded by
    \begin{align*}
        O(n^{1-g(\ell)}) \cdot \widetilde{O}(n^{5\sigma_{\ell-1}/\sigma_\ell+\sigma_{\ell-1}-\delta}) \leq \widetilde{O}(n^{1-g(\ell-1)-4\sigma_{\ell - 1}}).
    \end{align*}
    Therefore, the probability of having a $\ell$-spoiled vertex is at most $O(n^{1-g(\ell-1)-3\sigma_{\ell - 1}})$. On the other hand, since the expected number of $\ell$-spoiler vertices is less than 1, using a Chernoff bound we can show that with high probability there are at most $\widetilde{O}(1)$ $\ell$-spoiler vertices.
\end{proof}

\begin{claim}\label{clm:no-cycle-in-small-components}
    Let $C'_1, C'_2, \ldots, C'_{c'}$ be the connected components of black edges that do not contain any edge of $E^{inner}_\ell$. Then, with probability $1-O(n^{-\delta + \sigma_{\ell - 1} + 10\sigma_{\ell-1}/\sigma_\ell})$ all components are trees.
\end{claim}
\begin{proof}
    By \Cref{clm:small-component-in-intersect-tree}, with high probability $|C'_i| \leq O(n^{5\sigma_{\ell-1}/\sigma_\ell})$ for all $i \in [c']$. Also, $c' \leq O(n^{1-\delta + \sigma_{\ell - 1}})$ since the total number of black edges is $O(n^{1-\delta + \sigma_{\ell - 1}})$ by \Cref{clm:black-edges-count}. Hence, $\sum_{i=1}^{c'} |C_i'|^2 \leq O(n^{1-\delta + 2\sigma_{\ell - 1} + 10\sigma_{\ell-1}/\sigma_\ell})$.

    Suppose that we condition on high probability event that $\sum_{i=1}^{c'} |C_i'|^2 \leq O(n^{1-\delta + 2\sigma_{\ell - 1} + 10\sigma_{\ell-1}/\sigma_\ell})$. We add the edges of these components one by one with respect to the ordering that the algorithm queried. When the algorithm queries the adjacency list of vertex $v$ that is in a component of size $x$, the probability that the resulting queried edge goes to the same component is $O(x/n)$ by \Cref{clm:incoming-prob}. Therefore, if the edge is in the final connected component $C'_i$, this probability is upper bounded by $O(|C'_i|/n)$. Combining with the fact that each component has $|C'_i|$ edges, the probability of having a cycle is at most $\sum_{i=1}^{c'} O(|C_i'|^2/n) = O(n^{-\delta + \sigma_{\ell - 1} + 10\sigma_{\ell-1}/\sigma_\ell})$.
\end{proof}

For the rest, we condition on the event that each connected component of black edges that do not contain any edge of $E^{inner}_\ell$ is a tree. By \Cref{clm:no-cycle-in-small-components}, the failure probability of this event is $O(n^{-\delta + \sigma_{\ell - 1} + 10\sigma_{\ell-1}/\sigma_\ell}) = o(1)$. Moreover, since the number of levels in our hierarchy construction is a constant, these events hold for all levels with probability $1-o(1)$.

\begin{lemma}\label{lem:small-component-focs-modification}
    Let $(u,v)$ be a directed black edge in the connected component $C$ such that there is no edge of $E^{inner}_\ell$ in $C$. Also, suppose that $u \in A_r$ and $v$ belongs to $\{A_r, B_r, D_r\}$ in level $\ell-1$ of the hierarchy. Let $\overline{C}$ be the component that $v$ belongs to after removing edge $(u,v)$. Let $\mathcal{L}(v)$ and $\mathcal{L}'(v)$ be an arbitrary label for $v$ from $\{A_r, B_r, D_r\}$ and the entire queried subgraph of the black edges excluding $\overline{C}$. Then, we have 
    \begin{align*}
        \Pr[\overline{C} \mid \mathcal{L}(v)] \leq \left(1 + O(n^{\sigma_{\ell-1}-\delta}) \right)^{|\overline{C}|}\cdot \Pr[\overline{C} \mid \mathcal{L}'(v)].
    \end{align*}
\end{lemma}

We defer the proof of the above lemma to \Cref{sec:focs-modification2}.

\begin{lemma}\label{lem:large-component-focs-modification}
    Let $v$ be a vertex that is not $\ell$-spoiled and it belongs to a connected component with at least one edge of $E^{inner}_\ell$. Also, suppose that $v$ belongs to $\{A_r, B_r, D_r\}$ in level $\ell-1$ of the hierarchy. Let $\mathcal{L}(v)$ and $\mathcal{L}'(v)$ be an arbitrary label for $v$ from $\{A_r, B_r, D_r\}$ and the entire queried subgraph of the black edges excluding the $\ell$-shallow subgraph of $v$. Then, we have 
    \begin{align*}
        \Pr[T^\ell(v) \mid \mathcal{L}(v)] \leq \left(1 + O(n^{\sigma_{\ell-1}-\delta}) \right)^{|T^\ell(v)|}\cdot \Pr[T^\ell(v) \mid \mathcal{L}'(v)].
    \end{align*}
\end{lemma}

We defer the proof of the above lemma to \Cref{sec:focs-modification1}. Now we are ready to complete the proof of \Cref{lem:advantage-from-top}.

\begin{proof}[Proof of \Cref{lem:advantage-from-top}]
    As we discussed before, we need to prove \Cref{lem:advantage-from-top} for $\ell - 1$ using \Cref{lem:small-square-component} for $\ell$. Thus, $\ell > 1$. In this proof, when we use the label $A_r$ or $B_r$, we mean the $A_r$ and $B_r$ in level $\ell - 1$ of the hierarchy. The first part of the proof is similar to the proof of \Cref{lem:prob-of-edge-being-black}. Let $\widetilde{E}$
    be the set of black edges $(u, v)$ (directed from $u$ to $v$) such that $u\in A_r$ that satisfy at least one the following conditions:
    \begin{itemize}
        \item[(i)] $v$ is a $\ell$-spoiled vertex; or
        \item[(ii)] $u$ has at least $n^{\sigma_{\ell-1}}/3$ $\ell$-spoiled neighbors in the queried subgraph.
    \end{itemize}
    We begin by proving that for all black edges $e \notin \widetilde{E}$, we have that $p^{(\ell-1)-inner}_e \leq 10n^{\sigma_{\ell-2}-\sigma_{\ell-1}}$. Then, we give an upper bound on $|\widetilde{E}|$ with a case distinction based on the value of $g(\ell-1)$.

    Consider edge $e = (u,v)$ (directed from $u$ to $v$) where $e \notin \widetilde{E}$. If $u \notin A_r$, then the bound $p_e^{(\ell-1)-inner} = 0 \leq 10n^{\sigma_{\ell-2}-\sigma_{\ell-1}}$ trivially holds. So let us assume that $u \in A_r$. Let $v_0 = v, v_1, v_2, \ldots, v_k$ be the neighbors of $u$ that are adjacent to $u$ with a black edge in the queried subgraph such that $v_i \in A_r \cup B_r$ and either $v_i$ is a singleton vertex in the queried subgraph black edges or $v_i$ is the directed child of $u$ that is not spoiled. Note that $k\geq n^{\sigma_{\ell-1}}/2$ By condition (ii). We bound the probability that $v_0 \in A_r$ using a coupling argument.

    Consider a labeling profile $\mathcal{P}$ of all vertices $U = \{v_0, v_1, \ldots, v_k\}$ such that $\mathcal{P}(v_0) = A_r$. By the construction of our input distribution, since $u\in A_r$, at most $O(d_{\ell-2}) = O(n^{\sigma_{\ell-2}})$ vertices of $U$ are in $A_r$. We produce $\Omega(n^{\sigma_{\ell-1}})$ new profiles $\mathcal{P}'$ such that $\mathcal{P}'(v_0) \neq A_r$. For each vertex $v_i$ in $U$ such that $\mathcal{P}(v_i) = B_r$, we construct a new profile $\mathcal{P}'$ where $\mathcal{P}(v_j) = \mathcal{P}'(v_j)$ for $j \notin \{0, i\}$, $\mathcal{P}'(v_i) = A_r$, and $\mathcal{P}'(v_0) = B_r$. Since $v_0$ and $v_i$ are not a $\ell$-spoiled vertex, they are either in a component with no edge of $E^{inner}_\ell$ or their $\ell$-shallow subgraph satisfies the conditions in \Cref{def:spoiled-vertex-extend}. In both cases, the probability of querying the same shallow subgraph or connected component in the new labeling profile will be the same up to a factor of 
    $$\left(1 + O(n^{\sigma_L-\delta}) \right)^{n^{\delta - 2\sigma_{\ell-1}}},$$ by \Cref{lem:large-component-focs-modification} and \Cref{lem:small-component-focs-modification} since either $|T^\ell(v_0)| \leq n^{\delta-2\sigma_{\ell-1}}$  (resp. $|T^\ell(v_i)| \leq n^{\delta-2\sigma_{\ell-1}}$) or the component that $v_0$ or $v_i$ belongs to has size of at most $O(n^{5\sigma_{\ell-1}/\sigma_{\ell}}) \leq O(n^{\delta-2\sigma_{\ell-1}})$. Therefore, the probability of having profile $\mathcal{P}$ and $\mathcal{P'}$ are the same up to a factor $1 + o(1)$. We construct a bipartite graph $H = (P_1, P_2, E_P)$ of labeling profiles such that in $P_1$, we have all profiles $\mathcal{P}$ where $\mathcal{P}(v_0) = A_r$, and in the $P_2$, all profiles $\mathcal{P}'$ where $\mathcal{P}'(v_0) = B_r$. We add an edge between two profiles $\mathcal{P}$ and $\mathcal{P}'$ if we can convert $\mathcal{P}$ to $\mathcal{P}'$ according to the above process. Therefore, $\deg_H(\mathcal{P}) \geq k/2 \geq n^{\sigma_{\ell-1}}/4$ for $\mathcal{P} \in P_1$ since at least $k/2$ vertices of $U$ belong to $B_r$. On the other hand, $\deg_H(\mathcal{P}') \leq 2n^{\sigma_{\ell-2}}$ for $\mathcal{P}' \in P_2$. To see this, there are at most $2d_{\ell-2} = 2n^{\sigma_{\ell-2}}$ vertices $v_i$ in $U$ such that $\mathcal{P}'(v_i) = A_r$ according to the construction of input distribution. Hence,
    \begin{align*}
        p_{e}^{\ell-inner} &\leq (1 + o(1))\cdot \frac{|P_1|}{|P_2|}\\
        &\leq (1 + o(1)) \cdot \frac{2n^{\sigma_{\ell-2}}}{n^{\sigma_{\ell-1}}/4} \\
        &\leq (1 + o(1))\cdot 8n^{\sigma_{\ell-2} - \sigma_{\ell-1}} \leq 10n^{\sigma_{\ell-2} - \sigma_{\ell-1}}.
    \end{align*}
    Therefore, for all black edges $e \notin \widetilde{E}$, we have that $p^{\ell-inner}_e \leq 10n^{\sigma_{\ell-2}-\sigma_{\ell-1}}$. Now it remains to give an upper bound for $|\widetilde{E}|$. We prove this part using case distinction:

    \paragraph{(Case 1) $1-g(\ell-1) - 3\sigma_{\ell-1} \geq 0$:} first note that in this case  $1-g(\ell-1) > 0$, so we are in statement (ii) of \Cref{lem:advantage-from-top}. By \Cref{lem:ell-spoiled-count}, with high probability there are at most $O(n^{1-g(\ell - 1) - 2\sigma_{\ell-1}})$ $\ell$-spoiled vertices since $1-g(\ell-1) - 3\sigma_{\ell-1} \geq 0$. Further, each vertex has at most $\widetilde{O}(1)$ indegree which implies that there are at most $O(n^{1-g(\ell - 1) - \sigma_{\ell-1}})$ edges that satisfy condition (i). Now suppose that a vertex $u$ satisfies the condition (ii). Then, $u$ must have at least $n^{\sigma_{\ell-1}}$ edges $(u,w)$ (directed from $u$ to $w$) such that $w$ is $\ell$-spoiled since each vertex has at most $\widetilde{O}(1)$ indegree. Thus, the total number of vertices that satisfy condition (ii) is at most $\widetilde{O}(1)\cdot O(n^{1-g(\ell - 1) - 2\sigma_{\ell-1}}) \leq O(n^{1-g(\ell - 1) - \sigma_{\ell-1}})$. Therefore, we have $|\widetilde{E}| \leq O(n^{1-g(\ell - 1) - \sigma_{\ell-1}}) \leq O(n^{1-g(\ell - 1)})$ with high probability.

    \paragraph{(Case 2) $1-g(\ell-1) - 3\sigma_{\ell-1} < 0$ and $1-g(\ell - 1) > 0$:} in this case, since $1-g(\ell-1) - 3\sigma_{\ell-1} < 0$, by \Cref{lem:ell-spoiled-prob}, with high probability there are at most $\widetilde{O}(1)$ $\ell$-spoiled vertices which implies that $|\widetilde{E}| \leq  \widetilde{O}(1)$ with the same argument as case 1. Therefore, with high probability $|\widetilde{E}| \leq O(n^{1-g(\ell-1)})$.

    \paragraph{(Case 3) $1-g(\ell-1) - 3\sigma_{\ell-1} < 0$ and  $1-g(\ell - 1) < 0$:} in this case, since $1-g(\ell - 1) - 3\sigma_{\ell-1} < 0$, by \Cref{lem:ell-spoiled-prob}, with probability of $1-O(n^{1-g(\ell-1)-3\sigma_{\ell-1}}) \geq 1-O(n^{1-g(\ell-1)})$, there is no $\ell$-spoiled vertex which implies that $|\widetilde{E}| = 0$. Moreover, with high probability, there are at most $\widetilde{O}(1)$ $\ell$-spoiled vertices which implies that $|\widetilde{E}| = \widetilde{O}(1)$.
\end{proof}

\subsection{Proof of Lemma 6.15 and Lemma 6.42}\label{sec:focs-modification1}

In this section, we show our approach to proving \Cref{lem:focs-modification-coupling} and \Cref{lem:large-component-focs-modification}. Our proof draws inspiration from the findings of \cite{BehnezhadRR-FOCS23}. The way in which we construct each level of our input distribution closely resembles the hard example presented in this paper. The key distinction lies in how we put edges between different subsets of vertices. In their construction, they make an assumption that the degrees follow a binomial distribution. This assumption is beneficial because with each query the algorithm makes to a vertex's adjacency list, the neighbor's label becomes independent of the labels of the previously discovered neighbors. However, in order to maintain the condition of binomial degrees, they require a minimum of $O(\sqrt{n})$ {\em bad vertices}, where the neighbor distribution deviates from the expectation.  In their model, the total number of queries is significantly fewer than $O(\sqrt{n})$, allowing them to condition their process on not encountering any bad vertices. In contrast, in our setting, the algorithm can find one edge of at least $O(n^{1-\delta + \sigma_L})$ vertices which is much larger than $O(\sqrt{n})$. Therefore, we cannot expect not to see a bad vertex. So we slightly change their construction and use exact degrees between the subsets of vertices instead of binomial distribution. We will prove that the same result also holds in this construction. For now, suppose that we have a fixed level $\ell$ in our hierarchy. First, we introduced relevant notations and provided essential tools required to accomplish the final goal of this section. To provide a comprehensive overview, we reiterate certain definitions and claims as mentioned in \cite{BehnezhadRR-FOCS23}. For the rest of the section, assume that $d = d_\ell/d_{\ell-1} = \Theta(n^{\sigma_{\ell} - \sigma_{\ell-1}})$.

\begin{definition}[Special Edge][Similar to Definition 6.1 of \cite{BehnezhadRR-FOCS23}]\label{def:special-edge}
    We say an edge $(u,v)$ is special if one of the following statements holds:
    \begin{itemize}
        \item $u \in S$ and $v \in B_1$, or $u \in B_1$ and $v \in S$,
        \item $u \in B_i$ and $v \in A_{i-1}$, or $u \in A_{i-1}$ and $v \in B_i$ for $i \in (1, r]$,
        \item edges that only exist in $\yesdist^\ell$ or $\nodist^\ell$,
        \item edges between $D_i^j$ and $D_i^{j+1}$ for $j\in \{1, 3\}$ (for the base level we consider a perfect matching inside each $D_i$). 
    \end{itemize}
\end{definition}

\begin{definition}[Mixer Vertex][Similar to Definition 6.2 of \cite{BehnezhadRR-FOCS23}]\label{def:mixer-vertex}
Let $T$ be a rooted tree and $u$ be its root. Also, assume that $u \in \{A_{r}, B_{r}, D_{r} \}$. Let $v$ be a vertex in $T$ and suppose that there are $k$ special edges on the path between $u$ and $v$. If $k < r - 1$, we say $v$ is a mixer vertex if and only if $v \in \bigcup_{i=1}^{r-k-1} D_i$.
\end{definition}

The following observation is a direct consequence of \Cref{def:special-edge}, \Cref{def:mixer-vertex}, and the way the input distribution is constructed.

\begin{observation}\label{obs:required-special-edges}
    Let $T$ be a rooted tree where $u \in \{A_{r}, B_{r}, D_{r} \}$. Each path from $u$ to an $S$ vertex that does not contain a mixer vertex has at least $r - 1$ special edges.
\end{observation}

\begin{lemma}\label{lem:mixer-vertex-in-tree} Let $T$ be a rooted tree that is queried by the algorithm. Also, suppose that the root of the tree is in $\{A_r, B_r, D_r \}$. Then, with probability at least $1 - O(|V(T)|/d^{r-1})$, every path that starts from the root to an arbitrary vertex in the tree and does not contain a mixer vertex, must have at most $r - 2$ special edges.
\end{lemma}

\begin{proof}
    We prove that each path the algorithm finds to a vertex that contains at least $r-1$ special edges does not have a mixer vertex with probability $O(1/d^{r-1})$. For a mixer vertex in $D_i$ we use $i$ to denote the index of the mixer vertex. Suppose that there exists an oracle that each time the algorithm finds a path with at least $r-1$ special edges, it either returns that the path does not contain any mixer vertex or reveals the mixer vertex with the lowest index on the path. 

    Consider the first path that the algorithm finds with $r-1$ special edges. Consider the first time that the algorithm finds $r-2$ special edges on this path. Also, suppose that by this time, the path does not contain any $D_1$ vertex. Hence, by \Cref{obs:required-special-edges}, the path has not reached any vertex in layer 1 at this time. At this time, when the algorithm queries the next edge, the probability of seeing a special edge is $O(1/d)$ according to the construction. However, the probability of querying a vertex that is in $D_1$ is $\Theta(1)$. Therefore, the probability of the path going through the next special edge is $O(1/d)$ before stepping on a mixer vertex with index 1. The crucial difference between our construction and the construction in \cite{BehnezhadRR-FOCS23} appears here when for a fixed vertex $v$ if the oracle reveals a lot of mixer vertices that are direct children of $v$. Then, the probability of seeing a mixer vertex of level 1 when the algorithm queries the adjacency list of $v$ is not $\Omega(1)$ anymore. To deal with this issue, we give more power to the oracle. We assume that for vertex $v$, if the oracle revealed half of the mixer vertices of a fixed index $i$ that are direct children of $v$, the oracle reveals a path from $v$ downward to a vertex $w$ that does not contain a mixer with index $[1,i]$ and consequently it reveals the mixer of the path from the root to $w$ which must have an index larger than $i$. In the case that $i > r-2$, the oracle returns a path that does not contain any mixer vertex from the root, and the process terminates.

    With this modification, although the oracle gives more information, still we can get relatively the same result. Using the above argument, $O(1/d)$ fraction of paths do not cross a mixer vertex with index 1. Also, when the algorithm finds $\Omega(d)$ direct children of a vertex that are mixer vertices with index 1, the oracle gives away a path without having an index 1 mixer. Hence, the ratio of paths that the algorithm finds that do not contain a mixer vertex of index 1 is $O(1/d)$ fraction of all paths. Also, it is important to observe that when a mixer vertex $w$ is revealed by the oracle, all queries below that mixer vertex are pointless since the highest index mixer that will be revealed by the oracle for paths that cross $w$ is going to be $w$.

    Now consider all paths that do not contain a mixer vertex of index 1. With a similar argument, $O(1/d)$ fraction of these paths does not cross a mixer with index 2. To see this, the probability of crossing $(r-2)$-th special edge before going through a mixer with index 2 is $O(1/d)$. Therefore, $O(1/d^2)$ fraction of paths does not go through a mixer of index 2 or below. Similarly, the probability of having a path that does not cross any mixer vertex with an index of at most $i$ is $O(1/d^i)$. Therefore, the probability of having a path that does not contain any mixer vertex is $O(1/d^{r-1})$. Since the total number of paths from the root is at most $O(|V(T)|)$, with a probability of $1-O(|V(T)|/ d^{r-1})$ all paths that have more than $r-2$ special edges contain a mixer vertex.
\end{proof}

Note that the failure probability of the above event is very small. To see this, first, we have that $|V(T)| = O(n)$. Moreover, we have
\begin{align*}
    d^{r-1} \geq \Omega\left((n^{\sigma_\ell - \sigma_{\ell-1}})^{3r/4}\right) &\geq \Omega\left((n^{2\sigma_\ell/3})^{3r/4}\right) & (\text{Since } \sigma_\ell \geq (10/\delta)\cdot \sigma_{\ell-1})\\
    & = \Omega\left(n^{r \sigma_1 / 2}\right) & (\text{Since } \ell \geq 1 \text{ and } \sigma_i \geq \sigma_{i-1})\\
    & = \Omega\left( n^{5/\delta}\right) & (\text{Since } r\sigma_1 = 10/\delta)\\
    & = \Omega(n^5) & (\text{Since } \delta \leq 1).
\end{align*}
Therefore, the failure probability is $O(n^{-4})$, and using union bound, we can condition on the event that for all vertices that are not spoiled, the condition above holds. Now that we have this property, the exact same coupling as \cite{BehnezhadRR-FOCS23} works here since the number of neighbors of each subset of vertices is similar to the transition probabilities in their construction. We restate the lemma in terms of our parameter for both \Cref{lem:focs-modification-coupling} and \Cref{lem:large-component-focs-modification}.

\begin{lemma}[Similar to the Coupling Lemma in \cite{BehnezhadRR-FOCS23}. See Lemma 6.7 of the Arxiv version.]\label{lem:same-tree-different-labels}
    Let $T$ be a rooted tree that is queried by the algorithm where the root of the tree is in $\{A_r, B_r, D_r \}$. Also, suppose we condition on the event in \Cref{lem:mixer-vertex-in-tree}. Then, the probability of seeing the same tree is equal for all possible roots in $\{A_r, B_r, D_r \}$ up to $(1 + o(n^{2\delta - 3\sigma_L - 1}))^{|T|}$ multiplicative factor.
\end{lemma}

\begin{proof}[Proof of \Cref{lem:focs-modification-coupling}]
    First, by \Cref{obs:tree-structure-of-unspoiled}, since $v$ is not a spoiled vertex, the shallow subgraph of $v$ is a rooted tree. Note that we condition on the labels of all vertices except the vertices in the shallow subgraph of vertex $v$. However, in the coupling in \Cref{lem:same-tree-different-labels}, there is no conditioning on labels of vertices. Since the total number of the vertices that we are conditioning on their label is $O(n^{1-\delta + \sigma_{L}})$, the shift in probability of each step of the coupling in \Cref{lem:same-tree-different-labels} is at most $O(n^{1-\delta + \sigma_{L}} / n) = O(n^{ \sigma_{L} - \delta})$. On the other hand, the number of steps in coupling is $|T(v)|$, which implies that the total shift is upper bounded by 
\begin{align*}
\left(1 + o(n^{2\delta - 3\sigma_L - 1})\right)^{|T(v)|} \cdot \left(1 + O(n^{\sigma_{L}-\delta}) \right)^{|T(v)|} &\leq \left((1+o(1))\cdot O(n^{\sigma_{L}-\delta})\right)^{|T(v)|}\\
& \leq \left(1 + O(n^{\sigma_{L}-\delta})\right)^{|T(v)|},
\end{align*}

which concludes the proof.
\end{proof}

\begin{lemma}[Similar to the Coupling Lemma in \cite{BehnezhadRR-FOCS23}. See Lemma 6.7 of the Arxiv version.]\label{lem:same-tree-different-labels2}
    Let $T$ be a rooted tree with edges of level smaller than $\ell$ that is queried by the algorithm where the root of the tree is in $\{A_r, B_r, D_r \}$. Also, suppose we condition on the event in \Cref{lem:mixer-vertex-in-tree}. Then, the probability of seeing the same tree is equal for all possible roots in $\{A_r, B_r, D_r \}$ up to $(1 + o(n^{2\delta - 3\sigma_{\ell - 1} - 1}))^{|T|}$ multiplicative factor.
\end{lemma}

\begin{proof}[Proof of \Cref{lem:large-component-focs-modification}]
To begin, as per \Cref{obs:tree-structure-of-unspoiled-extension}, since $v$ is not an $\ell$-spoiled vertex, the $\ell$-shallow subgraph of $v$ forms a rooted tree. It is important to note that we condition our analysis on the labels of all vertices, excluding those in the $\ell$-shallow subgraph of vertex $v$. However, in the coupling detailed in \Cref{lem:same-tree-different-labels2}, there is no conditioning on vertex labels. Given that the total number of vertices for which we condition on their labels are at most $O(n^{1-\delta + \sigma_{\ell-1}})$, each step of the coupling in \Cref{lem:same-tree-different-labels2} has a probability shift of at most $O(n^{1-\delta + \sigma_{\ell-1}} / n) = O(n^{\sigma_{\ell-1} - \delta})$. On the other hand, the number of steps involved in the coupling process is $|T^\ell(v)|$, which implies that the total shift is upper bounded by
\begin{align*}
\left(1 + o(n^{2\delta - 3\sigma_{\ell-1} - 1})\right)^{|T^{\ell}(v)|} \cdot \left(1 + O(n^{\sigma_{\ell - 1}-\delta}) \right)^{|T^{\ell}(v)|} &\leq \left((1+o(1))\cdot O(n^{\sigma_{\ell-1}-\delta})\right)^{|T^{\ell}(v)|}\\
& \leq \left(1 + O(n^{\sigma_{\ell - 1}-\delta})\right)^{|T^{\ell}(v)|},
\end{align*}

which finishes the proof.
\end{proof}

\subsection{Proof of Lemma 6.41}\label{sec:focs-modification2}

In this section, we also employ analogous lemmas, such as \Cref{lem:mixer-vertex-in-tree} and \Cref{lem:same-tree-different-labels}, to show a coupling between the two distributions.

\begin{lemma}\label{lem:small-diameter-small-comp}
    Let $C$ be a connected component of black edges that is a tree such that there is no edge of $E^{inner}_\ell$ in $C$. With high probability, the longest path of the undirected edges of $C$ is smaller than $r-1$.
\end{lemma}
\begin{proof}
    Consider a path in component $C$ with length $k > r-2$. Suppose that we put the edges of the path on a line from left to right. Each edge has a direction that is either directed toward the left or directed toward the right. We let $a_i \in \{'\leftarrow', '\rightarrow' \}$ denote the direction of the edge on this line. Note that, by \Cref{cor:longest-black-directed-path}, the length of the longest directed path of black edges cannot be larger than $5/\sigma_\ell$. Thus, there must exist at least $k / (5/\sigma_\ell)$ different $i$ such that $a_i \neq a_{i+1}$ and $i < k$. For such $i$, we say that there is a collision at edge $i$. Furthermore, if there is a collision at edges $i_1$ and $i_2$ such that $i_2 > i_1$ and $i_2$ is the first collision after $i_1$, then $a_{i_1} \neq a_{i_2}$. Therefore, there must exist at least $\lfloor k/(10/\sigma_\ell) \rfloor > k/(20/\sigma_\ell)$ collisions such that $a_i = '\rightarrow'$ and $a_{i+1} = '\leftarrow'$. It is not hard to see that these types of collision are intersections between descendants of two vertices. Hence, if there exists a path of length $k$, then there must exist at least $k/(20/\sigma_\ell)$ intersections between descendants of vertices in component $C$.

    On the other hand, we have 
    \begin{align*}
        \frac{20k}{\sigma_\ell} & > \frac{10r}{\sigma_\ell} & (\text{Since } k > r/2)\\
        & = \frac{10^{L+1}}{\delta^{L+1}\cdot \sigma_\ell} & \left(\text{Since } r=\left(\frac{10}{\delta}\right)^{L+1}\right)\\
        & > 10 / \delta & (\text{Since } \sigma_\ell \geq 1 \text{ and } L \geq 0),
    \end{align*}
    which implies that there must exist more than $10/\delta$ intersections between descendants of vertices in component $C$ which is not possible because of \Cref{cor:small-component-in-intersect-tree-number}.
\end{proof}

\begin{corollary}\label{cor:mixer-vertex-in-small-components} Let $C$ be a connected component of black edges that is a tree such that there is no edge of $E^{inner}_\ell$ in $C$. Consider an arbitrary vertex $v$ in this component where $v \in \{A_r, B_r, D_r \}$. Then, all paths that start from $v$ to an arbitrary vertex in the component that does not contain a mixer vertex, have at most $r - 2$ special edges on it. 
\end{corollary}

\begin{proof}
    By \Cref{lem:small-diameter-small-comp} the longest path of the component $C$ is smaller than $r-1$ and therefore, no path in the component contains $r-1$ special edges. 
\end{proof}

Similar to the previous subsection, we can apply the same coupling as shown in Lemma 6.7 of \cite{BehnezhadRR-FOCS23}, as the number of neighbors for each subset of vertices aligns with the transition probabilities in their construction. Let us restate the lemma in the context of our parameters.

\begin{lemma}[Similar to the Coupling Lemma in \cite{BehnezhadRR-FOCS23}. See Lemma 6.7 of the Arxiv version.]\label{lem:same-comp-different-labels}
    Let $C$ be a connected component of black edges corresponding to the edges of level smaller than $\ell$ that is a tree such that there is no edge of $E^{inner}_\ell$ in $C$. Also, suppose that we condition on the event of \Cref{cor:mixer-vertex-in-small-components}. Then, the probability of seeing the same component is equal for both distributions up to $(1 + o(n^{2\delta - 3\sigma_{\ell - 1} - 1}))^{|C|}$ multiplicative factor.
\end{lemma}

\begin{proof}[Proof of \Cref{lem:small-component-focs-modification}]

Similar to the argument of proof of \Cref{lem:focs-modification-coupling} and \Cref{lem:large-component-focs-modification}, the total shift in the probability of the coupling is upper bounded by
\begin{align*}
\left(1 + o(n^{2\delta - 3\sigma_{\ell-1} - 1})\right)^{|\overline{C}|} \cdot \left(1 + O(n^{\sigma_{\ell - 1}-\delta}) \right)^{|\overline{C}|} &\leq \left((1+o(1))\cdot O(n^{\sigma_{\ell-1}-\delta})\right)^{|\overline{C}|}\\
& \leq \left(1 + O(n^{\sigma_{\ell - 1}-\delta})\right)^{|\overline{C}|},
\end{align*}
which yields the proof.  
\end{proof}

\section{Indistinguishability of Base Level Construction}\label{sec:finalizing-the-main-theorem}

In this section, first, we show that as a corollary of results in the previous section, we have $|E^{inner}_1| = 0$. This implies that the queried edges by the algorithm in the base level of our hierarchy create very small components, i.e. with size $O(n^{\sigma_1/\sigma_2})$. Moreover, we have the property that the union of these components is a forest and each connected component of the forest has a constant longest path. Then, we are able to use \Cref{lem:small-component-focs-modification} to show that the algorithm cannot distinguish if the base level construction is drawn from \yesdist{} or \nodist{} with probability $1-o(1)$.

\begin{corollary}\label{cor:no-detected-edge-base}
    With a probability of $1 - O(1/n)$, it holds $|E^{inner}_1| = 0$.
\end{corollary}
\begin{proof}
    Note that by statement (iii) of \Cref{obs:g-function-properties}, we have $g(1) > 2$. Thus, by \Cref{lem:advantage-from-top},
    \begin{align*}
        \Pr[E^{inner}_1 = \emptyset] \geq 1 - O(n^{1-g(1)}) \geq 1-O(1/n) = 1 - o(1). \qquad \qedhere
    \end{align*}
\end{proof}

\begin{claim}\label{clm:forest-base-level}
    With probability $1-o(1)$, all connected components of queried edges in the base level of the hierarchy are trees.
\end{claim}
\begin{proof}
    First, by \Cref{cor:no-detected-edge-base}, we have $|E^{inner}_1| = 0$ with probability $1-O(1/n)$. Let us condition on this event. Now, by \Cref{clm:no-cycle-in-small-components}, with probability $1-O(n^{-\delta + \sigma_1 + 10 \sigma_1/\sigma_2}) = 1-o(1)$, all connected components of queried edges in the base level of the hierarchy are trees which conclude the proof.
\end{proof}

The above claim enables us to use \Cref{lem:small-component-focs-modification} since all connected components are small and it is hard for the algorithm to learn the label of vertices in layer $r$ of the construction. This will help us to prove that the algorithm cannot distinguish if the base level of the construction is drawn from \yesdist{} or \nodist{}. We define {\em bad event} to be the event that \Cref{clm:forest-base-level} does not hold. By \Cref{clm:forest-base-level} the probability of the bad event is $o(1)$. Let us condition on not having a bad event. Now we prove that if there is no bad event in the queried subgraph of the base level of the hierarchy, then it is not possible for the algorithm to distinguish if the input graph is drawn from \yesdist{} or \nodist{}.

\begin{claim}\label{clm:low-degree-in-underlying}
    Let $V_B$ be the set of vertices that the algorithm finds at least one of their incident edges in the base level. Let $v \in V_B$ and $N_B(v)$ be all neighbors of $v$ in the queried subgraph of the base level. Then, with high probability, for each $v$ there are at most $\widetilde{O}(1)$ edges to vertices of $V_B \setminus N_B(v)$ in the underlying subgraph of base level.
\end{claim}
\begin{proof}
    We have $|V_B| = O(n^{1-\delta + \sigma_1})$ by \Cref{clm:black-edges-count}. Let $u \in V_B \setminus N_B(v)$. By \Cref{cor:edge-prob-bound}, the probability of having an edge between $v$ and $u$ is at most $O(n^{\sigma_L - 1})$. Define $X_u$ be the event that there exists an edge between $v$ and $u$. Thus, $\Pr[X_u = 1] \leq  O(n^{\sigma_L - 1})$. Let $X = \sum_{u \in V_B \setminus N_B(v)} X_u$. Hence, $\E[X] \leq O(n^{\delta + \sigma_1 + \sigma_L})$ because $|V_B \setminus N_B(v)| \leq O(n^{1-\delta + \sigma_1})$. Let $\lambda = (8\log n)/\E[X]$. Since events are negatively correlated, using the Chernoff bound we obtain
    \begin{align*}
        \Pr\left[X \geq (1+\lambda)\E[X] \right] &\leq \left(\frac{e^\lambda}{(1 + \lambda)^{1+\lambda}}\right)^{\E[X]}\\
        & \leq \left( \frac{e^\lambda}{\lambda^\lambda} \right)^{\E[X]} & (\text{Since } \lambda > 1)\\
        & = \left(\frac{e}{\lambda}\right)^{8\log n} &(\text{Since } \lambda = (8\log n)/\E[X])\\
        & \leq \frac{1}{n^8} &(\text{Since } \lambda > e^2).
    \end{align*}
    Therefore, with probability $1-n^{-8}$, there are at most $\widetilde{O}(1)$ edges to vertices of $V_B \setminus N_B(v)$ in the underlying subgraph of base level. Applying union bound for all vertices finishes the proof.
\end{proof}

\begin{lemma}\label{lem:same-outcome}
    Let us condition on not having the bad event defined above. Let $C_1, C_2, \ldots, C_c$ be the components of the forest that the algorithm found in the base level of the construction on a graph drawn from \yesdist{}. Then, the probability of querying the same forest in a graph that is drawn from \nodist{} is at least almost as large, up to  $1 + o(1)$ multiplicative factor.
\end{lemma}

\begin{proof}
By \Cref{lem:small-diameter-small-comp}, the maximum longest path of all components is smaller than $r-1$. Therefore, there is no path in any of the components that has $r-1$ special edges on it. We prove that the probability of seeing the same set of components is almost the same in both \yesdist{} and \nodist{} within $1 + o(1)$ multiplicative factor. Let $\mathcal{L}$ be the labeling in \yesdist{}. We will produce a labeling $\mathcal{L'}$ in \nodist{} and prove that the probability of seeing this labeling is almost the same as $\mathcal{L}$. With a similar approach, we can also couple each labeling in \nodist{} to a labeling in \yesdist{}. We start to iterate over the components one by one. Consider a component $C$. At any point, we condition on labels that we already revealed in $\mathcal{L'}$. If there is no edge between two vertices from $A_r$ in the component, we use the same labeling for $\mathcal{L'}$ since all other edges of \yesdist{} and \nodist{} are the same. 

Now suppose that there is an edge $(u,v)$  such that $u,v \in A_r$. Let $C_u$ and $C_v$ be two components that will be created if we remove edge $(u,v)$. In $\mathcal{L'}$, we let $u \in A_r$ and $v \in B_r$. We couple labels of $C_u$ and $C_v$ according to \Cref{lem:same-comp-different-labels}. We use the same approach as proof of \Cref{lem:small-component-focs-modification} and \Cref{lem:large-component-focs-modification}. In proof of \Cref{lem:same-comp-different-labels}, we assumed that because of the conditioning on revealed labels (in total $O(n^{ 1-\delta + \sigma_1})$ labels), there is an $O(n^{\sigma_1 - \delta})$ shift in the probability of the coupling of \Cref{lem:same-comp-different-labels} for each step of the coupling. However, this argument is loose, since each vertex in the component is connected to at most $\widetilde{O}(1)$ vertices with revealed labels by \Cref{clm:low-degree-in-underlying}. Conditioning on this fact, each step in the coupling is going to have at most $\widetilde{O}(1/n)$ shift in the probability, and in total we have $o(1)$ shift in the probability since the number of steps is equal to the total number of edges queried by the algorithm in the base level of construction which is $O(n^{1-\delta + \sigma_1})$. Therefore, we can couple the two distributions such that the probability of querying the same forest in both distributions is almost the same, up to  $1 + o(1)$ multiplicative factor.
\end{proof}

\begin{proof}[Proof of \Cref{lem:detlb}]
    By \Cref{lem:matching-size-final-graph}, any algorithm that estimates the size of the maximum matching with $\epsilon n$ additive error must be able to distinguish whether it belongs to \yesdist{} or \nodist{}. Furthermore, according to \Cref{lem:same-outcome}, the outcome distribution discovered by the algorithm is in a total variation distance of $o(1)$ for \yesdist{} and \nodist{}. Hence, the algorithm cannot between the support of two distributions with constant probability taken over the randomization of the input distribution. Therefore, any deterministic algorithm that provides an estimate $\widetilde{\mu}$ of the size of the maximum matching of $G$ such that $\E_G[\widetilde{\mu}] \geq \mu(G) - \epsilon n$ must spend at least $\Omega(n^{2-\delta})$ time.
\end{proof}

\paragraph{Acknowledgements.} Aviad Rubinstein is supported by David and Lucile Packard Fellowship.

\printbibliography

@article{saxenaKhursheed,
author = {Alam Khursheed and K. M. Lai Saxena},
title = {Positive dependence in multivariate distributions},
journal = {Communications in Statistics - Theory and Methods},
volume = {10},
number = {12},
pages = {1183-1196},
year = {1981},
publisher = {Taylor & Francis},
mydoi = {10.1080/03610928108828102},
myURL = { 
        https://doi.org/10.1080/03610928108828102
},
myeprint = { 
        https://doi.org/10.1080/03610928108828102
}

}

@article{kumarDevProschen,
 ISSN = {00905364},
 myURL = {http://www.jstor.org/stable/2240482},
 author = {Kumar Joag-Dev and Frank Proschan},
 journal = {The Annals of Statistics},
 number = {1},
 pages = {286--295},
 publisher = {Institute of Mathematical Statistics},
 title = {Negative Association of Random Variables with Applications},
 urldate = {2023-11-12},
 volume = {11},
 year = {1983}
}

@article{wajc2017negative,
  title={Negative association: definition, properties, and applications},
  author={Wajc, David},
  journal={Manuscript, available from https://goo. gl/j2ekqM},
  year={2017}
}

@article{Konig1916,
author = {König, D.},
journal = {Mathematische Annalen},
language = {ger},
pages = {453-465},
title = {Über Graphen und ihre Anwendung auf Determinantentheorie und Mengenlehre},
myurl = {http://eudml.org/doc/158740},
volume = {77},
year = {1916},
}

@inproceedings{behnezhad2021,
  author    = {Soheil Behnezhad},
  title     = {{Time-Optimal Sublinear Algorithms for Matching and Vertex Cover}},
  booktitle = {62nd {IEEE} Annual Symposium on Foundations of Computer Science, {FOCS}
               2021, Denver, CO, USA, February 7-10, 2022},
  pages     = {873--884},
  publisher = {{IEEE}},
  year      = {2021},
  myurl      = {https://doi.org/10.1109/FOCS52979.2021.00089},
  mydoi       = {10.1109/FOCS52979.2021.00089},
  bibsource = {dblp computer science bibliography, https://dblp.org}
}

@inproceedings{YoshidaYISTOC09,
  author    = {Yuichi Yoshida and
               Masaki Yamamoto and
               Hiro Ito},
  editor    = {Michael Mitzenmacher},
  title     = {An improved constant-time approximation algorithm for maximum matchings},
  booktitle = {Proceedings of the 41st Annual {ACM} Symposium on Theory of Computing,
               {STOC} 2009, Bethesda, MD, USA, May 31 - June 2, 2009},
  pages     = {225--234},
  publisher = {{ACM}},
  year      = {2009},
  myurl      = {https://doi.org/10.1145/1536414.1536447},
  mydoi       = {10.1145/1536414.1536447},
  bibsource = {dblp computer science bibliography, https://dblp.org}
}

@inproceedings{NguyenOnakFOCS08,
  author    = {Huy N. Nguyen and
               Krzysztof Onak},
  title     = {{Constant-Time Approximation Algorithms via Local Improvements}},
  booktitle = {49th Annual {IEEE} Symposium on Foundations of Computer Science, {FOCS}
               2008, October 25-28, 2008, Philadelphia, PA, {USA}},
  pages     = {327--336},
  year      = {2008},
  myurl      = {https://doi.org/10.1109/FOCS.2008.81},
  mydoi       = {10.1109/FOCS.2008.81},
  bibsource = {dblp computer science bibliography, https://dblp.org}
}

@inproceedings{KapralovSODA20,
  author    = {Michael Kapralov and
               Slobodan Mitrovic and
               Ashkan Norouzi{-}Fard and
               Jakab Tardos},
  title     = {{Space Efficient Approximation to Maximum Matching Size from Uniform
               Edge Samples}},
  booktitle = {Proceedings of the 2020 {ACM-SIAM} Symposium on Discrete Algorithms,
               {SODA} 2020, Salt Lake City, UT, USA, January 5-8, 2020},
  pages     = {1753--1772},
  year      = {2020},
  myurl      = {https://doi.org/10.1137/1.9781611975994.107},
  mydoi       = {10.1137/1.9781611975994.107},
  bibsource = {dblp computer science bibliography, https://dblp.org}
}

@inproceedings{ChenICALP20,
  author    = {Yu Chen and
               Sampath Kannan and
               Sanjeev Khanna},
  title     = {{Sublinear Algorithms and Lower Bounds for Metric {TSP} Cost Estimation}},
  booktitle = {47th International Colloquium on Automata, Languages, and Programming,
               {ICALP} 2020, July 8-11, 2020, Saarbr{\"{u}}cken, Germany (Virtual
               Conference)},
  pages     = {30:1--30:19},
  year      = {2020},
  myurl      = {https://doi.org/10.4230/LIPIcs.ICALP.2020.30},
  mydoi       = {10.4230/LIPIcs.ICALP.2020.30},
  bibsource = {dblp computer science bibliography, https://dblp.org}
}

@article{ParnasRon07,
  author    = {Michal Parnas and
               Dana Ron},
  title     = {{Approximating the Minimum Vertex Cover in Sublinear Time and a Connection
               to Distributed Algorithms}},
  journal   = {Theor. Comput. Sci.},
  volume    = {381},
  number    = {1-3},
  pages     = {183--196},
  year      = {2007},
  myurl      = {https://doi.org/10.1016/j.tcs.2007.04.040},
  mydoi       = {10.1016/j.tcs.2007.04.040},
  bibsource = {dblp computer science bibliography, https://dblp.org}
}

@inproceedings{Behnezhad-SODA23,
  author       = {Soheil Behnezhad},
  title        = {Dynamic Algorithms for Maximum Matching Size},
  booktitle    = {Proceedings of the 2023 {ACM-SIAM} Symposium on Discrete Algorithms,
                  {SODA} 2023, Florence, Italy, January 22-25, 2023},
  pages        = {129--162},
  year         = {2023},
  myurl          = {https://doi.org/10.1137/1.9781611977554.ch6},
  mydoi          = {10.1137/1.9781611977554.CH6},
  bibsource    = {dblp computer science bibliography, https://dblp.org}
}

@inproceedings{bhattacharyaKSW-SODA23,
  author       = {Sayan Bhattacharya and
                  Peter Kiss and
                  Thatchaphol Saranurak and
                  David Wajc},
  title        = {Dynamic Matching with Better-than-2 Approximation in Polylogarithmic
                  Update Time},
  booktitle    = {Proceedings of the 2023 {ACM-SIAM} Symposium on Discrete Algorithms,
                  {SODA} 2023, Florence, Italy, January 22-25, 2023},
  pages        = {100--128},
  year         = {2023},
  myurl          = {https://doi.org/10.1137/1.9781611977554.ch5},
  mydoi          = {10.1137/1.9781611977554.CH5},
  bibsource    = {dblp computer science bibliography, https://dblp.org}
}

@inproceedings{Yao77,
  author    = {Andrew Chi{-}Chih Yao},
  title     = {Probabilistic Computations: Toward a Unified Measure of Complexity
               (Extended Abstract)},
  booktitle = {18th Annual Symposium on Foundations of Computer Science, Providence,
               Rhode Island, USA, 31 October - 1 November 1977},
  pages     = {222--227},
  publisher = {{IEEE} Computer Society},
  year      = {1977},
  myurl       = {https://doi.org/10.1109/SFCS.1977.24},
  mydoi       = {10.1109/SFCS.1977.24},
  bibsource = {dblp computer science bibliography, https://dblp.org}
}

@inproceedings{BhattacharyaKS-STOC23,
  author    = {Sayan Bhattacharya and Peter Kiss and Thatchaphol Saranurak},
  title     = {Sublinear Algorithms for $(1.5 + \epsilon)$-Approximate Matching},
  year      = {2023},
  booktitle = {Proceedings of the 55th {ACM} Symposium on Theory of Computing, {STOC}
               2023, Orlando, Florida, to appear}
}

@inproceedings{Behnezhad-RRS-STOC23,
author = {Behnezhad, Soheil and Roghani, Mohammad and Rubinstein, Aviad},
title = {Sublinear Time Algorithms and Complexity of Approximate Maximum Matching},
year = {2023},
publisher = {Association for Computing Machinery},
address = {New York, NY, USA},
doi = {10.1145/3564246.3585231},
booktitle = {Proceedings of the 55th Annual ACM Symposium on Theory of Computing},
pages = {267–280},
numpages = {14},
location = {Orlando, FL, USA},
series = {STOC 2023}
}

@inproceedings{BehnezhadRRS-SODA23,
  author    = {Soheil Behnezhad and
               Mohammad Roghani and
               Aviad Rubinstein and
               Amin Saberi},
  editor    = {Nikhil Bansal and
               Viswanath Nagarajan},
  title     = {Beating Greedy Matching in Sublinear Time},
  booktitle = {Proceedings of the 2023 {ACM-SIAM} Symposium on Discrete Algorithms,
               {SODA} 2023, Florence, Italy, January 22-25, 2023},
  pages     = {3900--3945},
  publisher = {{SIAM}},
  year      = {2023},
  myurl       = {https://doi.org/10.1137/1.9781611977554.ch151},
  mydoi       = {10.1137/1.9781611977554.ch151},
  bibsource = {dblp computer science bibliography, https://dblp.org}
}

@inproceedings{BehnezhadRR-FOCS23,
  title={Local Computation Algorithms for Maximum Matching: New Lower Bounds},
  author={Behnezhad, Soheil and Roghani, Mohammad and Rubinstein, Aviad},
  booktitle={2023 IEEE 64th Annual Symposium on Foundations of Computer Science (FOCS)},
  pages={2322--2335},
  year={2023},
  organization={IEEE}
}

@inproceedings{BhattacharyaKS-FOCS23,
  author    = {Sayan Bhattacharya and
               Peter Kiss and
               Thatchaphol Saranurak},
  title     = {Dynamic {(1+{\(\epsilon\)})}-Approximate Matching
               Size in Truly Sublinear Update Time},
  booktitle = {64th {IEEE} Annual Symposium on Foundations of Computer Science, {FOCS}
               2023, to appear},
  year      = {2023},
  myurl       = {https://doi.org/10.48550/arXiv.2302.05030},
  mydoi       = {10.48550/ARXIV.2302.05030},
  eprinttype= {arXiv},
  eprint    = {2302.05030}
}

@inproceedings{MicaliV80,
  author       = {Silvio Micali and
                  Vijay V. Vazirani},
  title        = {{An $O(\sqrt(|V|) |E|)$ Algorithm
                  for Finding Maximum Matching in General Graphs}},
  booktitle    = {21st Annual Symposium on Foundations of Computer Science, Syracuse,
                  New York, USA, 13-15 October 1980},
  pages        = {17--27},
  year         = {1980},
  myurl          = {https://doi.org/10.1109/SFCS.1980.12},
  mydoi          = {10.1109/SFCS.1980.12},
  bibsource    = {dblp computer science bibliography, https://dblp.org}
}

@article{DuanP14,
  author       = {Ran Duan and
                  Seth Pettie},
  title        = {Linear-Time Approximation for Maximum Weight Matching},
  journal      = {J. {ACM}},
  volume       = {61},
  number       = {1},
  pages        = {1:1--1:23},
  year         = {2014},
  myurl          = {https://doi.org/10.1145/2529989},
  mydoi          = {10.1145/2529989},
  bibsource    = {dblp computer science bibliography, https://dblp.org}
}

@inproceedings{GoldreichR-STOC97,
  author       = {Oded Goldreich and
                  Dana Ron},
  title        = {Property Testing in Bounded Degree Graphs},
  booktitle    = {Proceedings of the Twenty-Ninth Annual {ACM} Symposium on the Theory
                  of Computing, El Paso, Texas, USA, May 4-6, 1997},
  pages        = {406--415},
  year         = {1997},
  myurl          = {https://doi.org/10.1145/258533.258627},
  mydoi          = {10.1145/258533.258627},
  bibsource    = {dblp computer science bibliography, https://dblp.org}
}

@inproceedings{GoldreichR-STOC98,
  author       = {Oded Goldreich and
                  Dana Ron},
  title        = {A Sublinear Bipartiteness Tester for Bunded Degree Graphs},
  booktitle    = {Proceedings of the Thirtieth Annual {ACM} Symposium on the Theory
                  of Computing, Dallas, Texas, USA, May 23-26, 1998},
  pages        = {289--298},
  year         = {1998},
  myurl          = {https://doi.org/10.1145/276698.276767},
  mydoi          = {10.1145/276698.276767},
  bibsource    = {dblp computer science bibliography, https://dblp.org}
}

@inproceedings{BKSpersonalcom,
  author       = {Sayan Bhattacharya and
                  Peter Kiss and
                  Thatchaphol Saranurak},
  title        = {Personal Communication},
  year         = {2023}
}

@inbook{AzarmehrBR-SODA24,
author = {Amir Azarmehr and Soheil Behnezhad and Mohammad Roghani},
title = {Fully Dynamic Matching: $(2-\sqrt{2})$-Approximation in Polylog Update Time},
booktitle = {Proceedings of the 2024 Annual ACM-SIAM Symposium on Discrete Algorithms (SODA)},
chapter = {},
pages = {3040-3061},
doi = {10.1137/1.9781611977912.109}
}

@article{ChenMetric-Arxiv22,
  author       = {Yu Chen and
                  Sanjeev Khanna and
                  Zihan Tan},
  editor       = {Kousha Etessami and
                  Uriel Feige and
                  Gabriele Puppis},
  title        = {Sublinear Algorithms and Lower Bounds for Estimating {MST} and {TSP}
                  Cost in General Metrics},
  journal    = {50th International Colloquium on Automata, Languages, and Programming,
                  {ICALP} 2023, July 10-14, 2023, Paderborn, Germany},
  series       = {LIPIcs},
  volume       = {261},
  pages        = {37:1--37:16},
  publisher    = {Schloss Dagstuhl - Leibniz-Zentrum f{\"{u}}r Informatik},
  year         = {2023},
  url          = {https://doi.org/10.4230/LIPIcs.ICALP.2023.37},
  doi          = {10.4230/LIPICS.ICALP.2023.37},
  bibsource    = {dblp computer science bibliography, https://dblp.org}
}

@article{TSP-icalp24,
  author       = {Soheil Behnezhad and
                  Mohammad Roghani and
                  Aviad Rubinstein and
                  Amin Saberi},
  title        = {Sublinear Algorithms for {TSP} via Path Covers},
  journal      = {CoRR},
  volume       = {abs/2301.05350},
  year         = {2023},
  url          = {https://doi.org/10.48550/arXiv.2301.05350},
  doi          = {10.48550/ARXIV.2301.05350},
  eprinttype    = {arXiv},
  eprint       = {2301.05350},
  bibsource    = {dblp computer science bibliography, https://dblp.org}
}
	
\end{document}